\documentclass[11pt]{article}
\usepackage{fullpage}
\usepackage{times}
\usepackage[numbers]{natbib}
\usepackage{subfigure}
\usepackage{graphicx,amssymb,amsmath}
\usepackage{epsfig}
\usepackage{algorithm}
\usepackage{latexsym}
\usepackage{url}
\usepackage{color}
\usepackage{graphicx,amssymb,amsmath}
\usepackage{url}
\newtheorem{theorem}{Theorem}[section]
\newtheorem{lemma}[theorem]{Lemma}

\newtheorem{defi}[theorem]{Definition}
\newtheorem{definition}[theorem]{Definition}

\newcommand{\sq}{\hbox{\rlap{$\sqcap$}$\sqcup$}}
\newcommand{\qed}{\hspace*{\fill}\sq}
\newenvironment{proof}{\noindent {\bf Proof.}\ }{\qed\par\vskip 4mm\par}
\author{
   Pingzhong Tang \\
   IIIS, Tsinghua University \\
   kenshinping@gmail.com
   \and
   Zihe Wang \\
   IIIS, Tsinghua University \\
   wzh5858588@163.com\\
   }\allowdisplaybreaks

\allowdisplaybreaks

\begin{document}


\title{Optimal Auctions for Negatively Correlated Items}
\author{
   Pingzhong Tang \\
   IIIS, Tsinghua University \\
   kenshinping@gmail.com
   \and
   Zihe Wang \\
   IIIS, Tsinghua University \\
   wzh5858588@163.com\\
   }

\maketitle
\begin{abstract}
We consider the problem of designing revenue-optimal auctions for selling two items and bidders' valuations are independent among bidders but negatively correlated among items. Abstractly, this setting can be thought of as an instance with single-dimensional type space but multi-dimensional allocation space. Such setting has been extensively studied in the literature, but all under the assumption that the items are positively correlated. Under the positive correlation assumption, the optimal allocation rules are {\em simple}, avoiding difficulties brought by ensuring Bayesian incentive compatibility (BIC) under multi-dimensional feasibility constraints and by ensuring interim individual rationality (IIR) given the possibility that the lowest utility point may no longer be at the boundary of the type domain. However, the nice properties no longer hold when there is negative correlation among items.

In this paper, we obtain the closed-form optimal auction for this setting, by directly addressing the two difficulties above. In particular, the first difficulty is that when pointwise maximizing virtual surplus under multi-dimensional feasibility (i.e., the Border feasibility), (1) neither the optimal interim allocation is trivially monotone in the virtual value, (2) nor the virtual value is monotone in the bidder's type. As a result, the optimal interim allocations resulting from virtual surplus maximization no longer guarantee BIC. To address (1), we prove a generalization of Border's theorem and show that optimal interim allocation is indeed monotone in the virtual value. To address (2), we adapt Myerson's ironing procedure to this setting by redefining the (ironed) virtual value as a function of the lowest utility point. The second difficulty, perhaps a more challenging one, is that the lowest utility type in general is no longer at the endpoints of the type interval. To address this difficulty, we show by construction that there exist an allocation rule and an induced lowest utility type such that they form a solution of the virtual surplus maximization and in the meanwhile guarantees IIR.

In the single bidder case, the optimal auction consists of a randomized bundle menu and a deterministic bundle menu; while in the multiple bidder case, the optimal auction is a randomization between two extreme mechanisms. The optimal solutions of our setting can be implemented by a Bayesian IC and IR auction, however, perhaps surprisingly, the revenue of this auction cannot be achieved by any (dominant-strategy) DIC and IR auction. In other words, {\em we have witnessed the first instance where a Bayesian IC auction yields strictly more revenue than any IC auction}, resolving an important conjecture concerning whether there exists a quasi-linear setting where the optimal revenues differ under the two different solution concepts.

\end{abstract}


\section{Introduction}


We consider a setting where a revenue-maximizing monopolist has two items for sale. The bidders' valuations are bidder-wise independent but negatively correlated item-wise, concretely in the form of $v_i^1+av_i^2=b$, where $v_i^1$ and $v_i^2$ are the private valuations of bidder $i$ for the two items respectively and $a,b$ are positive constant.

We obtain in closed-form revenue-optimal auctions subject to Bayesian incentive compatibility (BIC) and {\em interim} individual rationality (IIR) in this setting. For the case of a single bidder, the optimal auction consists of two options: a randomized bundle ($1/a$ fraction of the first item together with the second item) as well as a deterministic grand bundle. For the case of multiple i.i.d bidders, the optimal auction is a randomization between two extreme allocations resulted from solving a relaxed seller's problem. We show that such allocation can be implemented by a BIC and ({\em ex post}) IR auction.

Perhaps surprisingly, the optimal revenue achieved by the auction above cannot be achieved by any DIC and ({\em ex post}) IR auction. In other words, we have actually witnessed the first setting, to our best knowledge, where optimal revenues differ with respect to the two notions of IC, resolving a conjecture raised in a series of important papers concerning revenue equivalence between Bayes Nash implementation and dominant strategy implementation under quasi-linear utility \cite{manelli2010bayesian,gershkov2013equivalence,yao2015n}. Given the importance of this conjecture, we provide an independent proof for the case where two bidders have i.i.d types both drawn from uniform distribution on $[0,1]$.

To understand non-triviality of the problem and our technical contributions, let us first recall Myerson's approach for the single item case [Myerson 1981]. In Myerson's analysis, an important intermediate step is to show that the expected revenue equals the expected {\em virtual surplus}, minus the sum of bidders' utilities in their lowest types. By pointwise maximizing virtual surplus (i.e., give the item to the bidder with the highest non-negative virtual value), one can obtain an allocation rule whose interim expectation ({\em aka. allocation in the reduced-form})
\cite{cai2012optimal,cai2012algorithmic} is monotone in one's virtual value. Together with the regularity assumption (or the so-called ironing technique) that ensures the virtual value is further monotone in one's type, one can guarantee that the virtual surplus maximization allocation is monotone in bidder's interim type. In other words, the optimal allocation of the relaxed maximization problem (without considering IC and IR) satisfies IC.  Noticing that a bidder must attain his lowest utility at his lowest type, one can find, for any allocation rule, a payment rule (aka. payment identity) that guarantees zero utility for the lowest type agents. In other words, IR is ensured for any allocation rule with this payment identity.

In our setting, there are a number of essential differences.
The first observation is that the lowest utility point is not necessarily at the endpoints of the valuation interval.
It is also not easy to claim  there is a fixed zero point no matter what other bidders report.
As a result, a clean payment rule cannot be derived from an {\em ex post} type representation.
This difficulty motivates us to adopt an interim representation of the revenue formula, that is, to represent the revenue as a function of interim utility and interim allocation.

However, to accommodate interim allocation feasibility contraints in the problem formulation requires the so-called Border's theorem [Border 1992], a neat sufficient and necessary condition that characterizes exactly what set of interim allocation rules can actually be implemented by an auction ({\em ex post} allocation). Introducing Border's feasibility condition complicates the relaxed seller problem, i.e., the problem of choosing a virtual surplus maximization interim allocation without considering BIC and IIR. It is no longer trivial that the resulting allocation is monotone in virtual value, a key step in Myerson's analysis that guarantees BIC. Our first technical contribution is to ensure this is indeed the case. In particular, we put forth a surplus-maximization version of Border's theorem, that (1) remains to be a sufficient and necessary condition of the implementability of an interim allocation rule; (2) ensures the optimal allocation rule resulted from virtual surplus maximization (for any definition of virtual value function) subject to Border feasibility being monotone in the virtual value. With this generalization, together with a properly generalized ironing procedure (by defining virtual value as a function of both the actual type as well as the lowest utility type), we ensure that the allocation resulted from the relaxed seller's problem is indeed monotone in bidder's type.

A more challenging step is to guarantee IIR. Our strategy is to search through the set of allocations from the previous step (the solutions that solve the relaxed seller's problem and satisfies the conditions sufficient for the ironing procedure to guarantee optimal revenue) to find an allocation rule that ensures IIR and that the lowest utility is indeed zero. That is, this allocation rule achieves an allocation-independent revenue upper bound set by the relaxed seller's problem, while guarantees both BIC and IIR. In other words, this is the optimal allocation we look for. This step is established by a characterization of the set of all optimal allocations from the relaxed seller's problem, a careful analysis of possible allocations, and extreme allocations in particular, for the set of types that share the same virtual value and by establishing certain continuity of the optimal allocation function.

\subsection{Related work}

Negatively correlated valuations are not uncommon in the literature. For example, consider an instance of the well-known {facility location game} with two facilities \cite{lu2010asymptotically}: agents are interested in services from the two facilities, located at the endpoints of a street. Each agent's type is his/her location on the street, between the two facilities. The cost (negative utility) to either facility, the same as in \cite{procaccia2009approximate}, is simply one's distance to that facility, so his overall costs (as well as valuations) sum up to a constant --- the length of the street.

As mentioned, there have been extensive studies on extending Myerson's technique to settings with multi-dimensional flavor. \cite{haghpanah2015reverse,armstrong1996multiproduct} consider the aforementioned positive correlation setting, but restricted to the single bidder case and unit-demand valuations. In their setting, a bidder's type is represented by his valuation towards the more favored item and his valuation towards the less favored item is simply his type multiplied a discount factor, private in \cite{armstrong1996multiproduct} and public in \cite{haghpanah2015reverse}. To visualize, one can think of this type domain as a {\em ray} from the origin on the two-dimensional plane. In \cite{armstrong1996multiproduct}, it is shown that the optimal mechanism is simply to post a price for the higher valued item. Moreover, he shows that such price is independent from the discount factor, allowing rays with different slopes to have a uniform optimal mechanism. Since one can fully cover certain two dimensional distributions with such rays, the optimal posted price mechanism is actually optimal for these two-dimensional distributions as well. In \citep{haghpanah2015reverse}, the above result on rays is extended to curves satisfying certain monotone conditions. These conditions are derived by first assuming the optimality of the posted price mechanism and then inferring reversely under what type of distributions is the posted price mechanism virtual-surplus-maximizing. Comparing to these works, negative correlation adds additional difficulties in that the lowest utility point is not necessary at the boundary type. In addition, comparing to \cite{haghpanah2015reverse} that focuses on a well-behaved subset of distributions, we deal with the multi-dimensional feasibility directly without any assumption on the type distribution.

\cite{Levin97} considers a setting where the bidders valuations are positively correlated between the items and are both weakly increasing with respect to a one-dimensional type. Levin shows that, under regularity condition, this case can be solved using Myerson's very method. Our work can be seen as a complement of Levin's setting to the negative correlation setting.

Extending the ironing technique to multidimensional domain dates back to the seminal work by  \cite{rochet1998ironing}, where they study the single bidder case.
They invent a {\em sweeping procedure} that properly generalizes ironing to multidimensions. Their method is non-constructive, thus does not yield an explicit representation of the optimal mechanism.

At a higher level, our work is within the agenda of multidimensional revenue maximization.
In particular, our work is in the direction of exactly optimal mechanism design \cite{Daskalakis13,daskalakis2015strong,cai2011optimal,haghpanah2015reverse},
orthogonal to those aiming for approximate optimality  \cite{Hartline09,tang2011approximating,Hart2012a,tang2012mixed,Yao2013,yao2015n}.

There are a number of papers that consider the equivalence between BIC and DIC with respect to the objective of social welfare\cite{d1979incentives,cremer1988full}.
For the objective of revenue maximization,~\cite{manelli2010bayesian} proves that in the independent private values model with linear utility,
the outcome in terms of interim allocation rule and interim payment rule of any BIC mechanism, can also be obtained with a DIC mechanism (without any restriction on IR). Since the interim allocation rule and payment rule are same under the two IC notions, their work actually implies that any interim IR and BIC auction can be implemented in an interim IR and DIC auction. However, do notice that their result does not imply the equivalence between ex post IR and BIC mechanisms and ex post IR and DIC mechanisms, because equivalence in interim utilities does not imply equivalence in ex post utilities under the two notions of IC.

\cite{gershkov2013equivalence} extend the above work to general multi-dimensional, possibly non-linear settings. They provide a setting where, under a restricted definition of mechanism,
 the optimal DIC ``mechanism'' produces strictly less revenues than the optimal BIC ``mechanism''.

However, in their setting that sells two items to two bidders, any mechanism under their definition must be restricted to only three deterministic allocations: sell the bundle to the first bidder, sell the bundle to the second bidder, and sell one to each. So the revenue optimal ``mechanism'' under their definition is not optimal under the standard definition, because they preclude the possibility of reserving one or both of the items.

\section{Setting}

A seller has two items for sale. There are $n$ bidders interested in the items. For each bidder $i$, his valuation $(v_i^1,v_i^2)$ towards the two items are negatively correlated, satisfying $v_i^1+av_i^2=b$, where $a\geq 1$ and $b>0$. Now that $v_i^2=(b-v_i^1)/a$, we can use a one-dimensional variable $t_i=v_i^1$ to denote bidder $i$'s type. We assume each $t_i$ is independently, identically drawn from a differentiable density function $f$ on $[0,b]$. In other words, the probability density at valuation $(t_i, (b-t_i)/a)$ is $f(t_i)$.

By the revelation principle \cite{Myerson81}, it is without loss for the seller to focus on direct mechanisms. Such mechanism first solicits a type profile $\textbf{t}=(t_1,\ldots,t_n)$ from the bidders and then specifies an allocation rule and a payment rule. We use $q^j_i(\textbf{t})\geq 0, ~(j=1,2)$ to denote the probability that bidder $i$ is allocated the $j$-th item, and $pay_i(\textbf{t})$ to denote his payment. For each item $j$, the allocation rule must satisfy the feasibility constraint, i.e., $\sum_i  q^j_i(\textbf{t})\leq 1, \forall \textbf{t}$. Bidder $i$'s utility in this case is $u_i(\textbf{t})=q_i^1(\textbf{t})t_i+q_i^2(\textbf{t})(b-t_i)/a -pay_i(\textbf{t})$.

The seller's goal is to design a direct mechanism (i.e., auction) that maximizes the sum of payments (i.e., revenue) subject to {\em Bayesian incentive compatibility (BIC)} and {\em Interim individual rationality (IIR)} constraints~\cite{Myerson81}.

Bayesian incentive compatibility states that
no bidder has any incentive to lie about his valuation without knowing other bidders' types. Formally,
$$
E_{\textbf{t}_{-i}}[q^1_i(t'_i,\textbf{t}_{-i})t_i+q^2_i(t'_i,\textbf{t}_{-i})\frac{b-t_i}{a}-pay(t'_i,\textbf{t}_{-i})]
\leq
E_{\textbf{t}_{-i}}[u_i(t_i,\textbf{t}_{-i})]
$$

Interim individual rationality states that
a bidder never gets negative expected utility by participating in the auction. Formally, $E_{\textbf{t}_{-i}}[u_i(t_i,\textbf{t}_{-i})]\geq 0$.
Interim IR is a weaker concept than the standard {\em ex post} IR, which requires that a bidder never gets negative utility for any realization of type profile. Even though we have interim IR in the problem formulation, the optimal auctions we obtain are indeed {\em ex post} IR.

It is convenient to have the following {\em interim} notations (aka. reduced forms).

\begin{definition}
We use $q^j_i(t_i), ~i=1,2$ to denote the expected probability that bidder $i$ is allocated the $j$-th item at type $t_i$. Formally $q^j_i(t_i)=E_{\textbf{t}_{-i}}[q^j_i(t_i,\textbf{t}_{-i})]$.
Similarly, we can define interim payment and utility $pay_i(t_i)$ and $u_i(t_i)$.
\end{definition}

Furthermore, for bidders with i.i.d type distributions, by \citep{Maskin1984}, it is without loss of generality to only consider mechanisms that are symmetric in bidders.
In this case, we simplify notation by dropping subscript $i$, i.e. $q^j(t_i)=q^j_i(t_i)$, $pay(t_i)=pay_i(t_i)$ and $u(t_i)=u_i(t_i)$.
Furthermore we use $t$ to denote one bidder's type directly when there is no ambiguity.

It is well known \cite{rochet1985taxation, Hart2012b, wang2014optimal} that the BIC constraint is equivalent to that each bidder has a convex utility function, i.e., $u''(\cdot)\geq 0$.
On the other hand, given such a utility function, the allocation rule is given by the gradient of the utility function,
\begin{eqnarray}
u'(t)=\frac{\partial u}{\partial v^1}\frac{\partial v^1}{\partial t}(t)+\frac{\partial u}{\partial v^2}\frac{\partial v^2}{\partial t}(t)=q^1(t)-\frac{q^2(t)}{a}\nonumber
\end{eqnarray}

The seller's problem is formally defined as follows.
\begin{eqnarray}
Maximize && REV = n\int^b_0[t\cdot q^1(t)+\frac{b-t}{a}\cdot q^2(t)-u(t)]\cdot f(t) dt\nonumber\\
s.t. && \sum_i q^j_i(\textbf{t})\leq 1~ \forall \textbf{t}\in [0,b]^n,j=1,2\nonumber\\
&& u(t)\geq 0, u'(t)=q^1(t)-\frac{q^2(t)}{a}, u''(t)\geq 0, \forall t\in [0,b]\nonumber
\end{eqnarray}

\subsection{Border's theorem}

Note that in the seller's problem, everything is defined in terms of interim notations except for the feasibility constraints. Fortunately, the following theorem, well known as Border's theorem, allows us to equivalently convert the feasibility constraints into interim notations.
Theorem~\ref{border} and Lemma \ref{generalborder} are applied for the asymmetric case. We use $f_i$ to denote the density function of bidder $i$.

\begin{theorem}~\cite{border1991implementation,che2011generalized, daskalakis2015strong, gopalan2015public}
There exists a feasible mechanism that implements the interim allocation rule $q^j_i$ if and only if
$$\sum_i \int_{t\in S_i} q_i(t)f_i(t)dt\leq 1-\prod_i(1-\int_{t\in S_i}f_i(t)dt)~\forall S_i, i=1,...,n$$
$S_i$ is any subset of type of player $i$. Here, $S_i\subseteq[0,b], i=1...n$.
\label{border}
\end{theorem}

In a symmetric setting, feasible constraint reduces to that
the inequality always hold when $S_i, \forall i$ are the same.
By Border's Theorem, the seller's problem is equivalent to,

\begin{eqnarray}
Maximize && REV = n\int^b_0[t\cdot q^1(t)+\frac{b-t}{a}\cdot q^2(t)-u(t)]\cdot f(t) dt\nonumber\\
s.t. && \int_{t\in S} q^j(t)f(t)dt\leq \frac{1-(1-\int_{t\in S}f(t))^n}{n}~\forall S, j=1,2\nonumber\\
&& u(t)\geq 0, u'(t)=q^1(t)-\frac{q^2(t)}{a}, u''(t)\geq 0\nonumber
\end{eqnarray}

\subsection{Monotonicity of the optimal solution of the relaxed seller's problem}
Recall Myerson's techniques in the single-item case, revenue is expressed as the expected virtual surplus and the optimal allocation rule is to naturally allocate the item to the bidder with the highest virtual value. A key property, perhaps also a trivial one, is that the allocation rule is monotone with respect to the virtual value, i.e., higher virtual value implies higher probability of being allocated. This property, together with the regularity assumption (or the ironing technique in the irregular case), implies that the interim allocation is monotone in bidder's interim type. In other words, IC is guaranteed in the relaxed virtual surplus maximizing solution.

However, in our problem, such property is nontrivial because of the multi-dimensional feasibility constraints, i.e., Border feasibility. Therefore, to extend Myerson's technique to this setting, an important step is to show that, the virtual surplus maximizing allocation rule is still monotone with respect to the virtual value. Then, with a generalized ironing technique that ensures virtual value is monotone in interim type, we guarantee that the optimal allocation rule in the virtual surplus maximizing solution satisfies the IC constraints.

The following lemma and theorem, in the form of generalizations of Border's Theorem, are for this purpose. The lemma and theorem may be of independent interest.

%

\begin{lemma}
\label{generalborder}
There exists a feasible mechanism that implements the interim allocation rule $q_i$ if and only if
for any {\em virtual value function} $x^j_i(t)\in[0,+\infty)$,
$$\sum_{i}\int^{b}_0q_i(t)f_i(t)x_i(t)dt\leq \int^{+\infty}_{0}1-\prod_i(1-\int_{x_i(t)\geq v}f_i(t)dt)dv$$
\end{lemma}
This theorem subsumes Border's theorem. To see this, let $\overline{v}=1$, and
$x_i(t) = 1,  t\in S_i$, and $x_i(t) = 0,  t\notin S_i$.
\begin{proof}
If the interim allocation rule is implementable, we have
\begin{eqnarray}
\sum_i\int^b_{0}q_i(t)f_i(t)x_i(t)dt&=&\int^{+\infty}_0\sum_i\int_{x_i(t)\geq v} q_i(t) f_i(t) x_i(t)dtdv\nonumber\\
&\leq&\int^{+\infty}_0 1-\prod_i(1-\int_{x_i(t)\geq v}f_i(t)dt)dv\label{eq0}
\end{eqnarray}
The inequality follows from Border's theorem.
For the opposite direction,
since $q_i$ satisfies Border's condition by setting $x_i(t)=1$,
the interim allocation rule is implementable by Border's Theorem.
\end{proof}

From now on, we restrict to the i.i.d. bidders case.
The following theorem states that in the optimal solution of the {\em relaxed seller's problem} that only considers Border feasibility constraints, without IC and IR constraints\footnote{This problem is called the relaxed (seller's) problem throughout the paper.}, interim allocation probability $q(t)$ is monotone with respect to virtual value $x(t)$.
To formally state our result, the monotone property must be stated with respect to a nonzero measure.
\begin{theorem}
\label{monoborder}
In the i.i.d. bidders case, there exists a feasible mechanism that implements the interim allocation rule
$q$ if and only if
for any $x(t)\in[0,\overline{v}]$,
$$\int^{b}_0q(t)f(t)x(t)dt\leq \int^{\overline{v}}_{0}1-(1-\int_{x(t)\geq v}f(t)dt)^ndv$$
When the bound is achieved, for any set $\mathcal{C},\mathcal{D}$ with nonzero measure,
\begin{eqnarray}
\min_{t\in \mathcal{C}}\{x(t)\}> \max_{t\in \mathcal{D}}\{x(t)\}\Rightarrow \frac{\int_{t\in\mathcal{C}}q(t)f(t)dt}{\int_{t\in\mathcal{C}}f(t)dt} \geq \frac{\int_{t\in\mathcal{D}}q(t)f(t)dt}{\int_{t\in\mathcal{D}}f(t)dt}\nonumber
\end{eqnarray}
\end{theorem}



\section{Warm-up: the single bidder case}

In this section, we study the single bidder case. This case is easier because the feasibility constraint is simply $q^1(t),q^2(t)\in[0,1]$. Another important reason why this case is easier is that the lowest utility point is always at the lowest type (i.e., the lowest valuation of the first item), avoiding complexities brought by IR and by optimizing the location of the lowest utility point.

\begin{lemma}
In the optimal mechanism, $u(0)=0$.
\end{lemma}

\begin{proof}
If $u(t)>0, \forall t\in[0,b]$, we can always decrease the buyer utility function by simultaneously increasing the payment of every menu.
We assume $u(t^*)=0, t^*\in[0,b]$.

We first prove the payment when $t\in[0,t^*]$ must be less than or equal to $\frac{b}{a}$.
\begin{eqnarray}
pay(t)&=&tq^1(t)+\frac{b-t}{a}q^2(t)-u(t)\leq t(q^1(t)-\frac{1}{a}q^2(t))+\frac{b}{a}q^2(t)-0\nonumber\\
&\leq&0+\frac{b}{a}q^2(t)\leq\frac{b}{a}\nonumber
\end{eqnarray}
The first inequality follows from $u(t)\geq 0$.
The second inequality follows from the fact $q^1(t)-\frac{1}{a}q^2(t)=u'(t)\leq0$ when $t\in[0,t^*]$, since $u(t)$ is convex.
The last inequality follows from $q^2(t)\leq 1$.
The upper bound $b/a$ is achieved only when $(q^1(t),q^2(t),u(t))=(1/a,1,0)$.

For any $q^1(t),q^2(t)$ and $u(t)$, consider the following allocation and utility function:
\begin{displaymath}
(\overline{q}^1,\overline{q}^2,\overline{u}(t)) = \left\{ \begin{array}{ll}
(1/a,1,0) & t\in[0,t^*]\\
(q^1(t),q^2(t),u(t)) & t\in(t^*,b]
\end{array} \right.
\end{displaymath}

This is indeed achieved by a mechanism where the payment is $b/a,~t\in[0,t^*]$, which is the largest possible payment for this interval as proved earlier.
The payment remains the same for $t\in(t^*,b]$. As a result, the modified mechanism generates strictly more revenue than the original mechanism,
except when they are the same.

Furthermore, this modified mechanism clearly satisfies IR constraint by its definition and IR of the original mechanism.
We complete the proof by showing it is IC. Since $u(t^*)=0$ and $u(t)\geq 0,~t\in[0,b]$, we have $u'(t^*)=0$.
Since $u$ is convex, we have $u''(t)\geq 0,~t\in(t^*,b]$.
By definition, $u''(t)=0, t\in[0,t^*]$. Combined together $\overline{u}''(t)\geq 0,~t\in[0,b]$, i.e., IC is satisfied.
Thus the optimal mechanism must be in this form. In particular, $u(0)=0$.
\end{proof}

\subsection{The optimal solution}

Because $u(0)=0$, we can rewrite the revenue formula as follows,
\begin{eqnarray}
REV&=&\int^b_0[tq^1(t)+\frac{b-t}{a}q^2(t)-\int^{t}_0(q^1(s)-\frac{1}{a}q^2(s))ds-u(0)]f(t)dt\nonumber\\
&=&\int^b_0(q^1(t)-\frac{1}{a}q^2(t))(tf(t)-1+F(t))+\frac{b}{a}f(t)q^2(t)dt\nonumber
\end{eqnarray}
Let $h(t)=tf(t)+F(t)-1$ and apply the ironing technique of Myerson.
\begin{definition}\cite{Myerson81}
For any $z\in[0,1]$, let
$$H(z)=\int^{F^{-1}(z)}_{0}h(t)dt$$
and let $H^{ir}(z)$  be the convex hull of $H(z)$, i.e., the largest convex function that is less than or equal to $H(z)$.
We define $h^{ir}(t), t\in[0,b]$ such that $$H^{ir}(z)=\int^{F^{-1}(z)}_{0}h^{ir}(t)dt$$
\end{definition}

\begin{lemma}
The function $\frac{\partial H^{ir}}{\partial z}(z)$ is weakly increasing in $z$ and the ironed virtual value function $\frac{h^{ir}(t)}{f(t)}$ is weakly increasing in $t$.
\footnote{$\frac{\partial H^{ir}}{\partial z}(z)$ exists for almost every point except for some breaking points.}
\label{derivative}
\end{lemma}

It may be helpful to think of $\frac{h^{ir}(t)}{f(t)}$ as the {\em virtual value} of type $t$ in our setting. The above lemma claims that the virtual value is indeed monotone after ironing.

\begin{lemma}
We have $\int^b_0(q^1(t)-\frac{q^2(t)}{a})h(t)dt\leq \int^b_0(q^1(t)-\frac{q^2(t)}{a})h^{ir}(t)dt$.
If for almost every $t$,
we have either $H(F(t))=H^{ir}(F(t))$ or
$(q^1(t)-\frac{q^2(t)}{a})'=0$, then the inequality becomes equality.
\label{ironeq}
\end{lemma}

Similar to Myerson's analysis, the above lemma states that ironing does not hurt revenue, i.e., the ironed revenue achieves the revenue upper bound computed as if $h$ is monotone.

\begin{proof}
\begin{eqnarray}
&&\int^b_0(q^1(t)-\frac{q^2(t)}{a})h(t)dt=\int^b_0(q^1(t)-\frac{q^2(t)}{a})dH(F(t))\nonumber\\
&=&-\int^b_0(q^1(t)-\frac{q^2(t)}{a})'H(F(t)) dt + (q^1(t)-\frac{q^2(t)}{a})\int^{t}_{0}h(s)ds|^b_0\label{eq8}\\
&\leq&-\int^b_0(q^1(t)-\frac{q^2(t)}{a})'H^{ir}(F(t)) dt + (q^1(t)-\frac{q^2(t)}{a})H^{ir}(F(t))|^b_0\label{eq3}\\
&=&\int^b_0(q^1(t)-\frac{q^2(t)}{a})dH^{ir}(F(t))=\int^b_0(q^1(t)-\frac{q^2(t)}{a})h^{ir}(t)dt\nonumber
\end{eqnarray}
The first term in (\ref{eq8}) is less than or equal to the first term in (\ref{eq3}). The second terms in (\ref{eq8}) and (\ref{eq3}) are equal.
The inequality becomes equality when $(q^1(t)-\frac{q^2(t)}{a})'H(F(t))=(q^1(t)-\frac{q^2(t)}{a})'H^{ir}(F(t))$ fails on zero measure.
\end{proof}
Going back to the problem, it becomes
\begin{eqnarray}
Maximize && REV =\int^b_0(q^1(t)-\frac{1}{a}q^2(t))h(t)+\frac{b}{a}f(t)q^2(t)dt\nonumber\\
s.t. &&q^1(t),q^2(t)\in[0,1]\nonumber\\
&& q^1(t)-\frac{q^2(t)}{a}\geq 0,(q^1(t)-\frac{q^2(t)}{a})'\geq 0\nonumber
\end{eqnarray}
By Lemma~\ref{ironeq}, we have
$$REV\leq\int^{b}_{0}(q^1(t)-\frac{q^2(t)}{a})h^{ir}(t)dt+\frac{b}{a}f(t)q^2(t)dt$$
Let $s=\max_t\{t|h^{ir}(t)=0\}$, then $(q^1(t)-\frac{q^2(t)}{a})h^{ir}(t)\leq 0$, we get
\begin{eqnarray}
REV\leq\int^{b}_{s}(q^1(t)-\frac{q^2(t)}{a})h^{ir}(t)dt+\frac{b}{a}f(t)q^2(t)dt+\int^{s}_{0}\frac{bf(t)}{a}q^2(t)dt\label{eqb2}
\end{eqnarray}
Since $q^1(t)\leq 1, t\in[s,b]$ and $q^2(t)\leq 1, t\in[0,s]$, we have
\begin{eqnarray}
REV&\leq&\int^{b}_{s}(1-\frac{q^2(t)}{a})h^{ir}(t)+\frac{bf(t)}{a}q^2(t)dt+\int^{s}_{0}\frac{bf(t)}{a}dt\label{eqb1}\\
&=&\int^{b}_{s}h^{ir}(t)dt+\int^{s}_{0}\frac{bf(t)}{a}dt+\int^{b}_{s}\frac{bf(t)-h^{ir}(t)}{a}q^2(t)dt\nonumber
\end{eqnarray}
The coefficient of $q^2$ is nonnegative, since
\begin{eqnarray}
\frac{h^{ir}(t)}{f(t)}&=&\frac{\partial H^{ir}}{\partial t}(F(t))\leq\max_t{}\frac{\partial H^{ir}}{\partial t}(F(t))\leq\max_t{}\frac{\partial H}{\partial t}(F(t))\nonumber\\
&=&\max_t\frac{h(t)}{f(t)}=\max_{t}\{t+\frac{F(t)-1}{f(t)}\}\leq b\nonumber
\end{eqnarray}
Since $q^2(t)\in[0,1], t\in[s,b]$,
\begin{eqnarray}
REV\leq \int^{b}_{s}h^{ir}(t)dt+\int^{s}_{0}\frac{bf(t)}{a}dt+\int^{b}_{s}\frac{bf(t)-h^{ir}(t)}{a}dt\label{eqb3}
\end{eqnarray}
Construct $q^1$ and $q^2$ as follows,
\begin{displaymath}
(q^1(t),q^2(t)) = \left\{ \begin{array}{ll}
(1/a,1) & t\in[0,s]\\
(1,1) & t\in(s,b]
\end{array} \right.
\end{displaymath}
Using this allocation rule,  inequalities (\ref{eqb2}) (\ref{eqb1}) (\ref{eqb3}) become equalities.

Remember $H^{ir}(\cdot)$ is the convex hull of $H(\cdot)$.
If $H(F(t))>H^{ir}(F(t))$, curve $(z,H^{ir}(z)), z\in[F(t)-\epsilon, F(t)+\epsilon]$ must be a straight line for some $\epsilon$.
Thus $\partial \frac{h^{ir}(t)}{f(t)}/{\partial t}=0$.

When $ \frac{h^{ir}(t)}{f(t)}$ is equal for any two types $s_1$ and $s_2$,
$s_1$ and $s_2$ must be both larger than $s$, or both less than or equal to $s$.
By definition, $q^1(t)-\frac{q^2(t)}{a}$ is equal for the $s_1$ and $s_2$.
Thus the ironing inequality (\ref{eq3})  becomes an equality as well.

Then we get utility function and payment function:
\begin{displaymath}
u(t) = \left\{ \begin{array}{ll}
0, & t\in[0,s]\\
\frac{a-1}{a}(t-s), & t\in(s,b]
\end{array} \right.
\qquad
pay(t) = \left\{ \begin{array}{ll}
\frac{b}{a}, & t\in[0,s]\\
\frac{b}{a}+\frac{(a-1)s}{a}, & t\in(s,b]
\end{array} \right.
\end{displaymath}

\begin{theorem}
In the single-bidder case, the optimal mechanism consists of 2 menus:
the first is a randomized bundle $(1/a, 1, b/a)$ and the second is a deterministic bundle $(1, 1, \frac{b}{a}+\frac{(a-1)s}{a})$,
where $s=\max_t\{t|h^{ir}(t)=0\}$.
\end{theorem}
\section{The n-bidder case}

As mentioned, this case is more complicated for at least two reasons: Border feasibility, as well as the fact that the zero utility point may not be at the lowest type.
For any $t^*\in[0,b]$,
we first rewrite the revenue formula as a function of $t^*$ and $u(t^*)$.
Later we will pick $t^*$ such that it is the zero utility point.

\subsection{Revenue formula}
\begin{eqnarray}
\frac{REV}{n}&=&\int^b_0[t\cdot q^1(t)+\frac{b-t}{a}\cdot q^2(t)-u(0)-\int^t_0(q^1(s)-\frac{q^2(s)}{a})ds]\cdot f(t) dt\nonumber\\
&=&\int^b_0(q^1(t)-\frac{q^2(t)}{a})(tf(t)+F(t)-1)+q^2(t)\cdot\frac{bf(t)}{a}dt-u(0)\nonumber\\
&=&\int^b_0(q^1(t)-\frac{q^2(t)}{a})h(t,t^*)+q^2(t)\cdot\frac{bf(t)}{a}dt-u(t^*)\label{eq1}
\end{eqnarray}
Where,
\begin{displaymath}
h(t,t^*) = \left\{ \begin{array}{ll}
tf(t)+F(t) & t\leq t^*\\
tf(t)+F(t)-1 & t>t^*
\end{array} \right.
\end{displaymath}

\subsection{Ironing}

By separating $q^1$ and $q^2$, revenue formula (\ref{eq1}) can be rewritten as
$$\int^b_0 q^1(t)f(t)\frac{h(t,t^*)}{f(t)}+q^2(t)f(t)(\frac{b}{a}-\frac{h(t,t^*)}{af(t)})dt-u(t^*).$$
Now consider the relaxed maximization problem on the integral with only Border's feasibility.
Since for different $t^*$, the relaxed problem is different.
Let $q^1(t,t^*)$ and $q^2(t,t^*)$ denote the solution with respect to $t^*$.
Let $u(t,t^*)$ denote a still unfixed utility function while keeping $\frac{\partial u(t,t^*)}{\partial t}=q^1(t,t^*)-\frac{q^2(t,t^*)}{a}, t\in[0,b]$.
By Theorem \ref{monoborder}, in the optimal solution,
$q^1(t,t^*)$ and $-q^2(t,t^*)$ are monotone with respect to $\frac{h(t,t^*)}{f(t)}$.
Thus
$q^1(t,t^*)-\frac{q^2(t,t^*)}{a}$ is monotone with respect to $\frac{h(t,t^*)}{f(t)}$.
However $\frac{h(t,t^*)}{f(t)}$ might not be monotone in $t$,
then $q^1(t,t^*)-\frac{q^2(t,t^*)}{a}$ might not be monotone in $t$,
violating BIC.

To satisfy BIC, we hope the optimal solution of $q^1(t,t^*)$ and $-q^2(t,t^*)$ to be both monotone in $t$.
We apply a generalized ironing procedure to the coefficient function $h(t,t^*)$,
ensuring ironed virtual value $\frac{h^{ir}(t,t^*)}{f(t)}$ to be monotone in $t$.
After ironing, equation (\ref{eq1}) becomes
$\int^b_0 q^1(t)f(t)\frac{h^{ir}(t,t^*)}{f(t)}+q^2(t)f(t)(\frac{b}{a}-\frac{h^{ir}(t,t^*)}{af(t)})dt-u(t^*)$.
By Theorem~\ref{monoborder}, the upper bound is achieved when
$q^1(t,t^*)$ and $-q^2(t,t^*)$ are monotone with respect to $\frac{h^{ir}(t,t^*)}{f(t)}$.
Hence, $q^1(t,t^*)$ and $-q^2(t,t^*)$ are monotone in $t$ and BIC is satisfied.
When $\frac{h^{ir}(t,t^*)}{f(t)}$ is same, $q^1$, $q^2$ are not fixed, the monotonicity will be clarified in the next section.
\footnote{There is an exception when $\frac{h^{ir}(t,t^*)}{f(t)}$ is the same for several types, In this case, $q^1$ and $q^2$ are not fixed. How to determine $q^1$ and $q^2$ in this case and to prove monotonicity will be clarified in detail in the next section.}

Define the following ironing function with respect to any given type $t^*$.

\begin{definition}
For any $z\in[0,1]$, let
$$H(z,t^*)=\int^{F^{-1}(z)}_{0}h(t,t^*)dt$$ $H^{ir}(z,t^*)$ is the convex hull of  $H(z,t^*)$, i.e., the largest convex function that is less than or equal to $H(z,t^*)$.
We define $h^{ir}(s,t^*)$ such that $$H^{ir}(z,t^*)=\int^{F^{-1}(z)}_{0}h^{ir}(t,t^*)dt$$
\end{definition}

The following two lemmas are exactly the same as the single bidder case.
The second part of Lemma \ref{lemmaa1} states that, when certain conditions hold,
one can replace the virtual value function $h(t,t^*)$ by the ironed one $h^{ir}(t,t^*)$ in the revenue formula and achieve the same revenue upper bound.
However, in the n-bidder case, when the conditions hold, there may be many possible solutions of $q^1(t,t^*)$ and $q^2(t,t^*)$. In the next section, we will pick a specific pair of $q^1(t,t^*)$ and $q^2(t,t^*)$.

\begin{lemma}
\label{obs1}
The function $\frac{\partial H^{ir}}{\partial z}(z,t^*)$ is weakly increasing in $z$ and the virtual value function $\frac{h^{ir}(t,t^*)}{f(t)}$ is weakly increasing in $t$.
\end{lemma}

\begin{lemma}
We have $\int^b_0(q^1(t)-\frac{q^2(t)}{a})h(t,t^*)dt\leq \int^b_0(q^1(t)-\frac{q^2(t)}{a})h^{ir}(t,t^*)dt$.
Furthermore, if for almost every $t$ we have either $H(F(t),t^*)=H^{ir}(F(t),t^*)$ or
$(q^1(t)-\frac{q^2(t)}{a})'=0$, then the inequality becomes equality.
\label{lemmaa1}
\end{lemma}

With the two lemmas, by Equation (\ref{eq1}), we have
\begin{eqnarray}
\frac{REV}{n}&\leq& \int^b_0(q^1(t)-\frac{q^2(t)}{a})h^{ir}(t,t^*)+q^2(t)\frac{bf(t)}{a}dt-u(t^*)\nonumber\\
&=&\int^b_0q^1(t)h^{ir}(t,t^*)+q^2(t)\frac{bf(t)-h^{ir}(t,t^*)}{a}dt-u(t^*)\nonumber\\
&=&\int^b_0\int^{b}_{t}q^1(s)f(s)ds\frac{\partial(\frac{h^{ir}(t,t^*)}{f(t)})}{\partial t}dt
+\frac{1}{a} \int^b_0 \int^{t}_{0} q^2(s)f(s)ds\frac{\partial(\frac{h^{ir}(t,t^*)}{f(t)})}{\partial t}dt-u(t^*)\nonumber
\end{eqnarray}
By Lemma~\ref{obs1}, we have $\partial(\frac{h^{ir}(t,t^*)}{f(t)})/ \partial t\geq 0$. By Border's feasibility constraint, we have
$\int^{b}_{t} q^1(s)f(s)ds\leq \frac{1-F^n(t)}{n}$,
$\int^{t}_{0} q^2(s)f(s)ds\leq \frac{1-(1-F(t))^n}{n}$,
then

\begin{eqnarray}
\frac{REV}{n}&\leq& \frac{1-F^n(t)}{n}\frac{\partial(\frac{h^{ir}(t,t^*)}{f(t)})}{\partial t}dt+\frac{1}{a} \int^b_0 \frac{1-(1-F(t))^n}{n}\frac{\partial(\frac{h^{ir}(t,t^*)}{f(t)})}{\partial t}dt-u(t^*)\label{eq2}\\
&\leq& \frac{1-F^n(t)}{n}\frac{\partial(\frac{h^{ir}(t,t^*)}{f(t)})}{\partial t}dt+\frac{1}{a} \int^b_0 \frac{1-(1-F(t))^n}{n}\frac{\partial(\frac{h^{ir}(t,t^*)}{f(t)})}{\partial t}dt
\label{eq4}
\end{eqnarray}

Up to now, we have shown that, for any $t^*$, there exist allocation functions $q^1(t,t^*)$ and $q^2(t,t^*)$ such that (\ref{eq2}) becomes an equation (thus, by Theorem \ref{monoborder}, such allocation function guarantees BIC).

It remains to prove there exist $t^*$, together with $q^1(t,t^*)$, $q^2(t,t^*)$ and $u(t,t^*)$ found in the first step, such that $u(t^*,t^*)=0$, $t^*$ is the lowest point of $u(t,t^*)$ and (\ref{eq4}) becomes an equation, i.e. such $t^*$ together with $q^1(t,t^*)$, $q^2(t,t^*)$ and $u(t,t^*)$ ensure IR.

To sum up we show that there is indeed an IC, IR mechanism that achieves the allocation-independent revenue upper bound (\ref{eq4}) set by Border feasibility alone, in which the lowest utility point is at type $t^*$.

\section{Satisfying IR}

We first state our main result of this section.

\begin{theorem}
There exists a type $t^*$, allocation rule $q^1(t,t^*)$, $q^2(t,t^*)$ and utility function $u(t,t^*)$
such that inequality (\ref{eq4}) becomes an equality, also $u(t,t^*)\geq 0$ and $\frac{\partial u}{\partial t}(t,t^*)=q^1(t,t^*)-\frac{q^2(t,t^*)}{a}$.
\label{thm1}
\end{theorem}

The proof of Theorem \ref{thm1} is divided into several steps. First, we provide a structural characterization of the set of $q^1(t^*,t^*)$ and $q^2(t^*,t^*)$ from all possible solutions of the relaxed problem. It turns out that, for each $t^*$, there exists a closed interval (could be a single point) such that any number in this interval corresponds to such a $q^1(t^*,t^*)$,
Similarly for $q^2(t^*,t^*)$ and $q^1(t^*,t^*)-\frac{q^2(t^*,t^*)}{a}$.

Second, we show that, such value intervals of $q^1(t^*,t^*)$ for each $t^*$ seamlessly connect, as $t^*$ increases.
Then the value intervals of $q^1(t^*,t^*)-\frac{q^2(t^*,t^*)}{a}$ seamlessly connect, as $t^*$ increases.
We then choose $t^*$ based on properties of the union of these intervals.

\subsection{Finalizing the optimal allocation rule}
To determine the value interval of $q^1(t^*,t^*)$ for each $t^*$, we need to investigate a set of types around $t^*$ such that, for each type $t$ in this set,
$h^{ir}(t,t^*)/f(t)=h^{ir}(t^*,t^*)/f(t^*)$.
In other words, these types have the same ironed virtual value as at type $t^*$. However, these types may have different interim allocation probabilities, i.e., $q^1(t,t^*)$ for these types may be different. This is also different from the Myerson case where virtual value is the unique identity to determine the interim allocation rule.
In our case, the set of types with the same ironed virtual value can be partitioned into several subsets.
The partition points are exactly those where the convex hull $H^{ir}$ touches $H$.

It can be shown that, the largest (rightmost point of the interval) value of $q^1(t^*,t^*)$ is given by uniformly setting the highest possible interim probability for some subsets and
the lowest interim probability on the remaining subsets. The smallest (leftmost point of the interval) value of $q^1(t^*,t^*)$ is given similarly.

To make the above ideas formal, we now introduce several notations regarding how to partition the set of types with the same ironed virtual value.

\begin{lemma}
$H(F(t^*),t^*)> H^{ir}(F(t^*),t^*)$, when $t^*\in(0,b)$.
\label{abovelemma}
\end{lemma}
It says point $(t^*, H(F(t^*),t^*))$ is above point $(t^*, H^{ir}(F(t^*),t^*))$ when $t^*\in(0,b)$.
Thus $l^{min}_2(t^*,t^*)>l^{max}_1(t^*,t^*)$. By this lemma, $(\hat{q}^1,\hat{q}^2)$ and $(\check{q}^1,\check{q}^2)$ are well defined.

\begin{defi} For any $t^*\in(0,b)$, we define
\begin{eqnarray}
l^{min}_1(t,t^*)&=&\min_{x\leq F(t)}(\frac{\partial H^{ir}}{\partial z}(x,t^*)=\frac{\partial H^{ir}}{\partial z}(F(t),t^*))\nonumber\\
l^{min}_2(t,t^*)&=&\min_{x\geq F(t)}(H^{ir}(x,t^*)=H^{ir}(F(t),t^*))\nonumber\\
l^{max}_1(t,t^*)&=&\max_{x\leq F(t)}(H^{ir}(x,t^*)=H^{ir}(F(t),t^*))\nonumber\\
l^{max}_2(t,t^*)&=&\max_{x\geq F(t)}(\frac{\partial H^{ir}}{\partial z}(x,t^*)=\frac{\partial H^{ir}}{\partial z}(F(t),t^*))\nonumber
\end{eqnarray}
\label{defi1}
\end{defi}
\begin{figure}
  \centering
  \includegraphics[width=4.5cm]{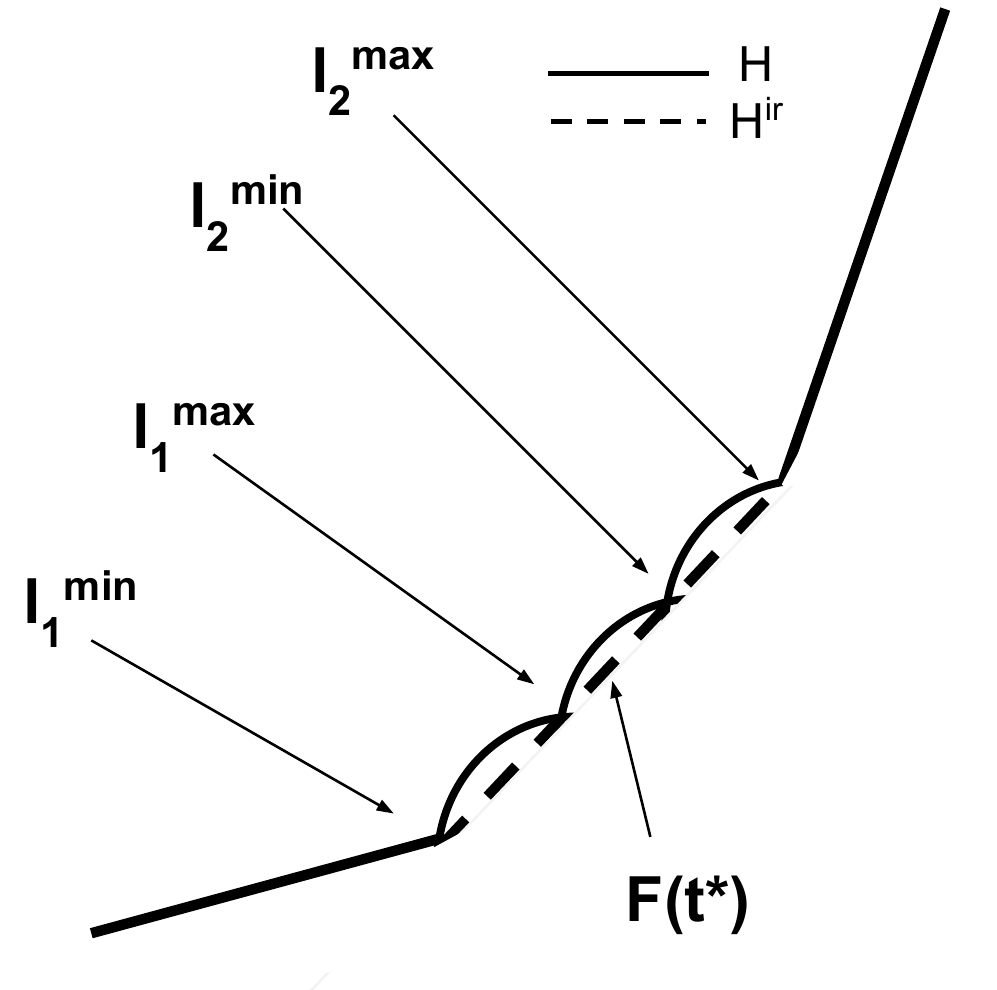}\quad
  \includegraphics[width=5.5cm]{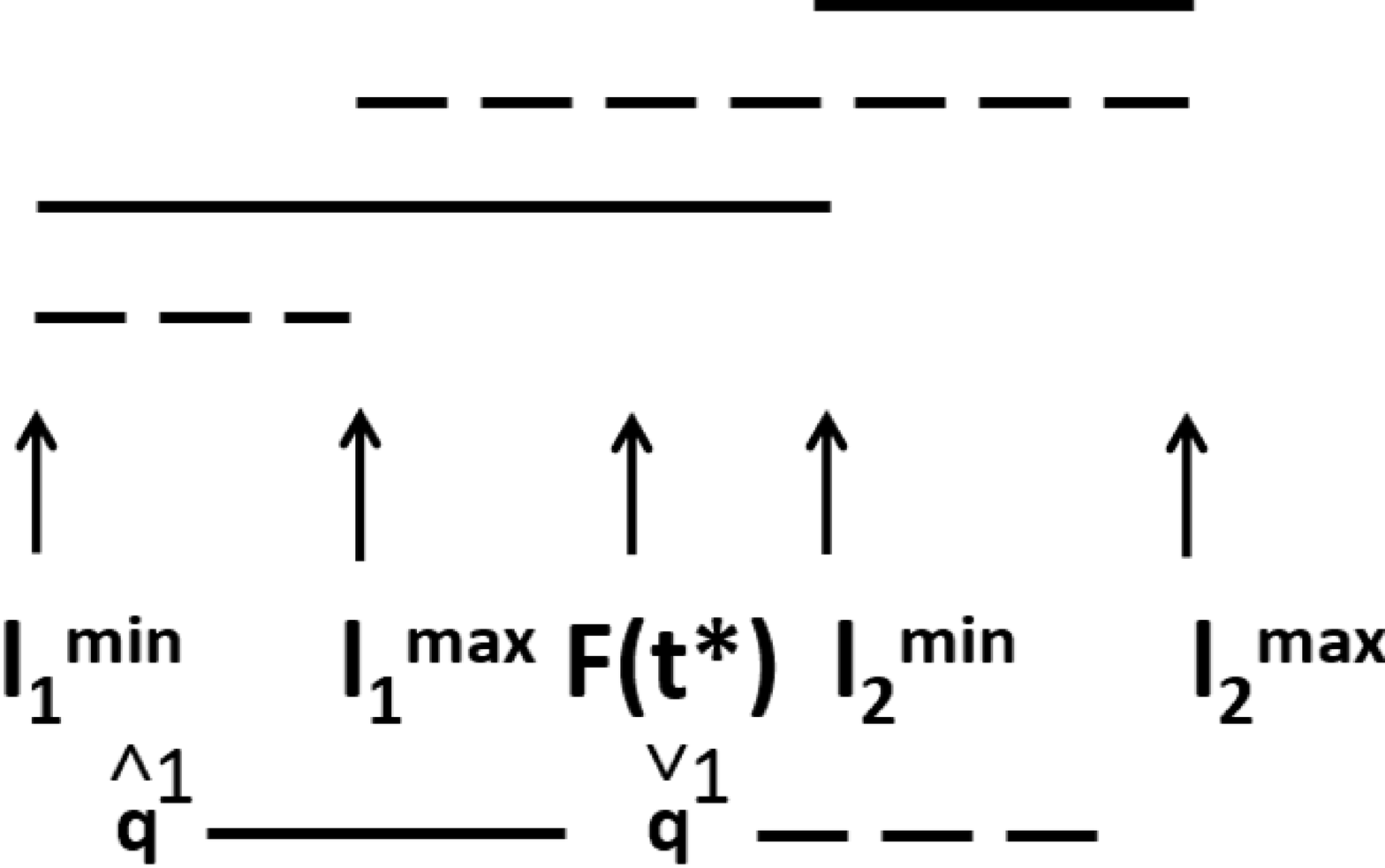}\quad
  \caption{Partition points and allocations}\label{fig1}
\end{figure}
Look at Fig \ref{fig1}, $l^{min}_1(t,t^*)$ is the leftmost point of the straight line segment
 on which point $(t,H^{ir}(t,t^*))$ lies. Similar for $l^{min}_2(t,t^*)$, $l^{max}_1(t,t^*)$, $l^{max}_2(t,t^*)$.
Next we define the interim allocation rule.

If $(\frac{h^{ir}(t,t^*)}{f(t)})'>0$, $H^{ir}(F(t),t^*)=H(F(t),t^*)$.
$q^1(t,t^*)=F^{n-1}(t)$, $q^2(t,t^*)=(1-F(t))^n$.

If $(\frac{h^{ir}(t,t^*)}{f(t)})'=0$, there are two cases:
First, if $\frac{h^{ir}(t,t^*)}{f(t)}=\frac{h^{ir}(t^*,t^*)}{f(t)}$, there are two extreme allocations.
One extreme allocation is to assign a higher probability on $[F^{-1}(l^{min}_2(t^*,t^*)),F^{-1}(l^{max}_2(t^*,t^*))]$ and a lower probability
on $[F^{-1}(l^{min}_1(t^*,t^*)),F^{-1}(l^{min}_2(t,t^*)))$.
\begin{displaymath}
\hat{q}^1(t,t^*) = \left\{ \begin{array}{ll}
\frac{1}{n}\frac{(l_2^{min}(t^*,t^*))^n-(l_1^{min}(t^*,t^*))^n}{l_2^{min}(t^*,t^*)-l_1^{min}(t^*,t^*)}, & F(t)\in [l^{min}_1(t^*,t^*),l^{min}_2(t^*,t^*))\\
\frac{1}{n}\frac{(l_2^{max}(t^*,t^*))^n-(l_2^{min}(t^*,t^*))^n}{l_2^{max}(t^*,t^*)-l_2^{min}(t^*,t^*)}, & F(t)\in [l^{min}_2(t^*,t^*),l^{max}_2(t^*,t^*)]
\end{array} \right.
\end{displaymath}
\begin{displaymath}
\hat{q}^2(t,t^*) = \left\{ \begin{array}{ll}
\frac{1}{n}\frac{(1-l_1^{min}(t^*,t^*))^n-(1-l_2^{min}(t^*,t^*))^n}{l_2^{min}(t^*,t^*)-l_1^{min}(t^*,t^*)}, & F(t)\in [l^{min}_1(t^*,t^*),l^{min}_2(t^*,t^*))\\
\frac{1}{n}\frac{(l_2^{min}(t^*,t^*))^n-(1-l_2^{max}(t^*,t^*))^n}{l_2^{max}(t^*,t^*)-l_2^{min}(t^*,t^*)}, & F(t)\in [l^{min}_2(t^*,t^*),l^{max}_2(t^*,t^*)]
\end{array} \right.
\end{displaymath}

The other extreme allocation is to assign a higher probability on $[l^{max}_1(t^*,t^*),l^{max}_2(t^*,t^*)]$ and a lower one on
$[l^{min}_1(t^*,t^*),l^{max}_1(t,t^*))$. That is,
\begin{displaymath}
\check{q}^1(t,t^*) = \left\{ \begin{array}{ll}
\frac{1}{n}\frac{(l_2^{min}(t^*,t^*))^n-(l_1^{min}(t^*,t^*))^n}{l_2^{min}(t^*,t^*)-l_1^{min}(t^*,t^*)}, & F(t)\in [l^{min}_1(t^*,t^*),l^{max}_1(t^*,t^*))\\
\frac{1}{n}\frac{(l_2^{max}(t^*,t^*))^n-(l_2^{min}(t^*,t^*))^n}{l_2^{max}(t^*,t^*)-l_2^{min}(t^*,t^*)}, & F(t)\in [l^{max}_1(t^*,t^*),l^{max}_2(t^*,t^*)]
\end{array} \right.
\end{displaymath}
\begin{displaymath}
\check{q}^2(t,t^*) = \left\{ \begin{array}{ll}
\frac{1}{n}\frac{(1-l_1^{min}(t^*,t^*))^n-(1-l_2^{min}(t^*,t^*))^n}{l_2^{min}(t^*,t^*)-l_1^{min}(t^*,t^*)}, & F(t)\in [l^{min}_1(t^*,t^*),l^{max}_1(t^*,t^*))\\
\frac{1}{n}\frac{(l_2^{min}(t^*,t^*))^n-(1-l_2^{max}(t^*,t^*))^n}{l_2^{max}(t^*,t^*)-l_2^{min}(t^*,t^*)}, & F(t)\in [l^{max}_1(t^*,t^*),l^{max}_2(t^*,t^*)]
\end{array} \right.
\end{displaymath}

Second, if $\frac{h^{ir}(t,t^*)}{f(t)}\neq \frac{h^{ir}(t^*,t^*)}{f(t)}$, we assign the same interim allocation probability to all the types with same ironed virtual value.
\begin{displaymath}
q^1(s,t^*) = \left\{ \begin{array}{ll}
\frac{1}{n}\frac{(l_2^{max}(t,t^*))^n-(l_1^{min}(t,t^*))^n}{l_2^{max}(t,t^*)-l_1^{min}(t,t^*)}, & F(s)\in [l^{min}_1(t,t^*),l^{max}_2(t,t^*))
\end{array} \right.
\end{displaymath}
\begin{displaymath}
q^2(s,t^*) = \left\{ \begin{array}{ll}
\frac{1}{n}\frac{(1-l_1^{min}(t,t^*))^n-(1-l_2^{max}(t,t^*))^n}{l_2^{max}(t,t^*)-l_1^{min}(t,t^*)}, & F(s)\in [l^{min}_1(t,t^*),l^{max}_2(t,t^*))
\end{array} \right.
\end{displaymath}
Let $\alpha$ denote a parameter in $[0,1]$.
Consider
$q^1(t,t^*)= \alpha\hat{q}^1(t,t^*)+(1-\alpha)\check{q}^1(t,t^*)$ and
$q^2(t,t^*)= \alpha\hat{q}^2(t,t^*)+(1-\alpha)\check{q}^2(t,t^*), t\in[l_1^{min}(t^*,t^*),l_2^{max}(t^*,t^*)]$.
When $H^{ir}(F(t),t^*)<H(F(t),t^*)$, we always have $\partial(q^1(t,t^*)-\frac{q^2(t,t^*)}{a})/\partial t=0$.
By Lemma \ref{lemmaa1}, $q^1(t,t^*)$ and $q^2(t,t^*)$ still maximize the relaxed problem.
$q^1(t,t^*)$ and $-q^2(t,t^*)$ are both monotone in $t$, which guarantees BIC.
When $\alpha$ goes from 0 to 1, $q^1(t^*,t^*)$ can achieve any value between the interval $[\check{q}^1(t^*,t^*),\hat{q}^1(t^*,t^*)]$.
Furthermore, $q^1(t^*,t^*)-\frac{q^2(t^*,t^*)}{a}$ can achieve any value in the corresponding interval.

We give example of $q^1$ when $t^*=0$. It is similar when for $t^*=b$ and $q^2$.
The difference between Definition~\ref{defi1} is that we assign a unique interim allocation probability
$q^1(t,0)=[l^{max}_2(0,0)]^{n-1}$ when $t\in[0,l^{max}_2(0,0)]$. $q^1(t,0)$ is same as in Definition~\ref{defi1}.
Then, as promised, we show the value intervals of $q^1(t^*,t^*)$ are seamlessly connect as $t^*$ increases.
Let $[k_1(t^*),k_2(t^*)]$ denote the value interval for $q^1(t^*,t^*)$.
We say intervals seamlessly connect,
if for any two types $t^{*1}<t^{*2}$, we have that $k_2(t^{*1})\leq k_1(t^{*2})$,
and $\bigcup_{t^*}[k_1(t^*),k_2(t^*)]$ is connected.

\begin{lemma}
The value intervals of $q^1(t^*,t^*)$ seamlessly connect as $t^*$ increases.
\label{biglemma}
\end{lemma}
Finally we can prove Theorem~\ref{thm1}
\begin{proof}
By Lemma~\ref{biglemma}, we know the value intervals of $q^1(t^*,t^*)-\frac{q^2(t^*,t^*)}{a}$ also seamlessly connect.
If zero is in the union of value intervals, we assume it is in the value interval of $q^1(t^{**},t^{**})-\frac{q^2(t^{**},t^{**})}{a}$.
Choose $\alpha$ such that $q^1(t^{**},t^{**})= \alpha\hat{q}^1(t^{**},t^{**})+(1-\alpha)\check{q}^1(t^{**},t^{**})$.
Let $q^1(t,t^{**})= \alpha\hat{q}^1(t,t^{**})+(1-\alpha)\check{q}^1(t,t^{**})$ for $t\in [l^{min}_1(t^{**},t^{**}),l^{max}_2(t^{**},t^{**}))$.
Similar for $q^2(t,t^{**})$.

Let $u(t^{**},t^{**})=0$ and $\frac{\partial u}{\partial t}(t,t^{**})=q^1(t,t^{**})-\frac{q^2(t,t^{**})}{a}$.
Since $\frac{\partial u}{\partial t}(t^{**},t^{**})=q^1(t^{**},t^{**})-\frac{q^2(t^{**},t^{**})}{a}=0$ and $\frac{\partial u}{\partial t}(t,t^{**})$ is monotone in $t$, we have
$u'(t,t^{**})\leq 0,~t\in[0,t^{**}]$ and $u'(t,t^{**})\geq 0,~t\in[t^{**},b]$. Then $u(t^{**},t^{**})$ is indeed the lowest utility point.
The constructed $q^1(\cdot,t^{**})$, $q^2(\cdot,t^{**})$ and $u(\cdot,t^{**})$ satisfies both IIR and BIC and achieves the revenue upper bound in Equation (\ref{eq4}) and is therefore revenue optimal.

If zero is below(above) the union of value intervals, we set $t^{**}=0(b)$ and $\alpha=0(1)$ instead.
The zero utility point will be $0(b)$.
\end{proof}

\subsection{Implementation: the form of the optimal auction}
Given type $t^*$, interim allocation rule $q^1(t,t^*)$ and $q^2(t,t^*)$, and expected utility function $u(t,t^*)$ that in Theorem \ref{thm1}.
We now describe the form of the optimal mechanism.


The optimal auction first computes and selects the bidders with the highest ironed virtual value and compare this virtual value to a constant $\frac{h^{ir}(t^*,t^*)}{f(t^*)}$.
    \begin{itemize}
    \item If the two numbers are not equal, then the first item is uniformly allocated to the highest ironed virtual value bidders; the second item is uniformly allocated to he lowest ironed virtual value bidders.
    \item If the two values are equal, we randomize between two mechanisms.
    Let $\alpha$ denote a parameter such that $q^1(t^*,t^*)= \alpha\hat{q}^1(t^*,t^*)+(1-\alpha)\check{q}^1(t^*,t^*)$.
        With probability $\alpha$ (a constant that depends on $t^*$), we run the following Mechanism 1:
        \begin{itemize}
        \item If there exist types in some interval $[l_2^{min}(t^*,t^*),l_2^{max}(t^*,t^*)]$, item one is uniformly allocated to the bidders in this interval.
        \item If there does not exist such a type, item one is uniformly allocated to the highest ironed virtual value bidders.
        \item The second item is allocated in the same way but with a different interval condition $[l_1^{min}(t^*,t^*),l_2^{min}(t^*,t^*)]$.
        \item In Mechanism 1, when the ironed virtual value is equal to $\frac{h^{ir}(t^*,t^*)}{f(t^*)}$, there are priorities for types in $[l_2^{min}(t^*,t^*),l_2^{max}(t^*,t^*)]$ for getting item 1, there are priorities for types in $[l_1^{min}(t^*,t^*),l_2^{min}(t^*,t^*)]$ for getting item 2.
        \end{itemize}
        With probability $(1-\alpha)$, we run Mechanism 2. Mechanism 2 is identical to Mechanism 1 except for different intervals.
    \end{itemize}

Roughly speaking, the auction needs to deal with the case where there are multiple highest ironed virtual value bidders and needs to allocate among them delicately to implement the optimal interim allocation rule.

Let $(\dot{q}^1(\textbf{t}),\dot{q}^2(\textbf{t}),\dot{pay}(\textbf{t}))$ denote the mechanism we construct.
In fact $E_{\textbf{t}_{-i}}\dot{q}^1(t_i,\textbf{t}_{-i})=q^1(t_1,t^*)]$. The payment rule is:
$$pay_i(t_i,\textbf{t}_{-i})=(1-\frac{u(t_i)}{t_iq^1(t_i,t^*)+\frac{b-t_i}{a}q^2(t_i,t^*)})(t_i\dot{q}^1(t_i,\textbf{t}_{-i})+\frac{b-t_i}{a}\dot{q}^2(t_i,\textbf{t}_{-i}))$$
This guarantees that $\dot{u}(t_i,\textbf{t}_{-i})$ is always nonnegative, and we have
\begin{eqnarray}
E_{\textbf{t}_{-i}}\dot{u}(t_i,\textbf{t}_{-i})&=&\frac{u(t_i)}{t_iq^1(t_i,t^*)+\frac{b-t_i}{a}q^2(t_i,t^*)}E_{\textbf{t}_{-i}}(t_i\dot{q}^1(t_i,\textbf{t}_{-i})+\frac{b-t_i}{a}\dot{q}^2(t_i,\textbf{t}_{-i}))=u(t_i,t^*)\nonumber
\end{eqnarray}
By Theorem~\ref{thm1}, $(\dot{q}^1(\textbf{t}),\dot{q}^2(\textbf{t}),\dot{pay}(\textbf{t}))$ is indeed an optimal IR and BIC mechanism.
By \cite{manelli2010bayesian}, there is a DIC and IIR mechanism that achieves the same optimal revenue.
However this mechanism in general does not satisfy DIC and {\em ex post} IR at the same time.
We will discuss this later in more details via an example.

\section{An example where BIC differs from DIC in revenue, under {\em ex post IR}}

\subsection{Two bidders with uniform distributions}

There are two i.i.d. bidders with uniform valuation $v^1+v^2=1$, i.e. $F_i(t_i)=F(t_i)=t_i,~t_i\in[0,1]$.
If there is one such bidder, the optimal mechanism is trivial, i.e., to sell the items in a bundle at price 1. In the two-bidder case, first of all,
consider optimal BIC and IIR mechanism.
WLOG, we can restrict attention to bidder-symmetric mechanism \cite{Maskin1984}.
Note that if a mechanism is DIC and {\em ex post} IR, then there is also a symmetric mechanism which is DIC and {\em ex post} IR.

By our approach, one can show that $t^*$ must equal $\frac{1}{2}$ in the optimal BIC and IIR mechanism.
In the optimal symmetric BIC and IIR mechanism, we first prove the allocation rule and the expected utility are fixed in certain subset of types.
From Equation (\ref{eq1}), we have
\begin{eqnarray}
\frac{REV}{2}&=&\int^{1}_{0}q^1(t)h(t,\frac{1}{2})+q^2(t)(1-h(t,\frac{1}{2}))dt-u(\frac{1}{2})\nonumber\\
&\leq& \int^{1}_{0}(q^1(t)-q^2(t))h(t,\frac{1}{2})+q^2(t)dt\label{egeq1}\\
&\leq&\int^{1}_{0}(q^1(t)-q^2(t))h^{ir}(t,\frac{1}{2})+q^2(t)dt\label{egeq2}\\
&=& \int^{1}_{0}q^1(t)h^{ir}(t,\frac{1}{2})+q^2(t)(1-h^{ir}(t,\frac{1}{2}))dt\nonumber
\end{eqnarray}
\begin{displaymath}
\textrm{where,}\quad h(t,\frac{1}{2}) = \left\{ \begin{array}{ll}
2t & t\in[0,1/2]\\
2t-1 & t\in(1/2,1]
\end{array} \right. ,\quad h^{ir}(t) = \left\{ \begin{array}{ll}
2t & t\in[0,1/4)\\
1/2 & t\in[1/4,3/4]\\
2t-1 & t\in(3/4,1]
\end{array} \right.
\end{displaymath}

\begin{displaymath}
\frac{REV}{2}=\int_{[0,1/4]\cup[3/4,1]}[\int^1_{t}q^1(s)ds+\int^{t}_{0}q^2(s)ds]dt
\end{displaymath}
By Border's Theorem
\begin{displaymath}
\int^1_{t}q^1(s)ds\leq \frac{1-t^2}{2} \textrm{~and~} \int^{t}_0q^2(s)ds \leq \frac{1-(1-t)^2}{2}
\end{displaymath}
We have
\begin{eqnarray}
&\frac{REV}{2}\leq \int_{[0,1/4]\cup[3/4,1]} \frac{1-t^2}{2} +\frac{1-(1-t)^2}{2}dt=\frac{29}{48}&\label{egeq3}
\end{eqnarray}
To attain the revenue upper bound  $\frac{29}{24}$,
(\ref{egeq1})(\ref{egeq2}) and (\ref{egeq3}) must be equalities.

\begin{itemize}
\item Equality (\ref{egeq1}) requires $u(\frac{1}{2})=0$.
\item Equality (\ref{egeq2}) requires $q^1(t)-q^2(t)=0,~t\in(\frac{1}{4},\frac{3}{4})$.
\item Equality (\ref{egeq3}) requires $\int^1_{t}q^1(s)ds=\frac{1-t^2}{2}$, $\int^t_{0}q^1(s)ds=\frac{1-(1-t)^2}{2}$ for
almost every $t\notin[\frac{1}{4},\frac{3}{4}]$.
This further implies when the lower type is in $[0,\frac{1}{4}]\cup[\frac{3}{4},1]$, the item is allocated to the bidder with a higher type.
In particular, $q^1(t_i,0)=1$ and $q^2(t_i,0)=0$ for $t\in(0,1]$.
Interim allocation rule in some intervals is also fixed,
\begin{displaymath}
 \left\{ \begin{array}{ll}
q^1(t,\frac{1}{2})=t & t\in[0,1/4]\cup [3/4,1]\\
q^2(t,\frac{1}{2})=1-t & t\in[0,1/4]\cup [3/4,1]
\end{array} \right.
\end{displaymath}
However, interim allocation rule in interval $[1/4,3/4]$ is not fully fixed.
The expected utility function is fixed
\begin{displaymath}
u(t,\frac{1}{2}) = \left\{ \begin{array}{ll}
\frac{3}{16}-t+t^2 & t\in[0,1/4]\cup [3/4,1]\\
0 & t\in(1/4,3/4)
\end{array} \right.
\end{displaymath}

\end{itemize}

We now show that there is a BIC and {\em ex post} IR mechanism that achieves the revenue bound.
In the following, we construct the {\em ex post} IR and BIC mechanism.
The ex post allocation rule is
\begin{displaymath}
q^1_1(t_1,t_2) = \left\{ \begin{array}{ll}
1/2 & t_1=t_2 \textrm{~or~} (t_1,t_2)\in(1/4,3/4)\times(1/4,3/4)\\
1& t_1>t_2 \textrm{~and~} (t_1,t_2)\notin(1/4,3/4)\times(1/4,3/4)\\
0 & o.w.
\end{array} \right.
\end{displaymath}
Let $q^2_1(t_1,t_2)=q^1_2(t_1,t_2)=1-q^2_1(t_1,t_2)$ and $q^2_2(t_1,t_2)=q^1_1(t_1,t_2)$.

To satisfy the ex post IR constraint, the payment rule is a little tricky.
We let it be the valuation multiplies a factor $\beta(t_1)\in[0,1]$, i.e.,
\begin{displaymath}
pay_1(t_1,t_2)=\beta(t_1)(q^1_1(t_1,t_2)t_1+q^2_1(t_1,t_2)(1-t_1))
\end{displaymath}
\begin{displaymath}
\textrm{where,~}\beta(t)=\frac{u(t)}{q^1(t)t+q^2(t)(1-t)}=\left\{ \begin{array}{ll}
\frac{3/16-t+t^2}{t^2+(1-t)^2} & t\in[0,1/4]\cup [3/4,1]\\
0 & t\in(1/4,3/4)
\end{array} \right.
\end{displaymath}
It is easy to verify that $u_1(t_1,t_2)\geq 0$ and $E_{t_2}u_1(t_1,t_2)=u(t_1)$.

To see this mechanism is not DIC, consider the case when $t_2=0$.
When $t_1\in (0,1]$, by definition we have $(q^1_1(t_1,0),q^2_1(t_1,0))=(1,0)$, and
\begin{displaymath}
pay(t_1,0) = (1-\beta(t_1))t_1 =\left\{ \begin{array}{ll}
\frac{3t_1/16-t^2_1+t^3_1}{t_1^2+(1-t_1)^2} & t_1\in[0,1/4]\cup [3/4,1]\\
0 & t_1\in(1/4,3/4)
\end{array} \right.
\end{displaymath}
For $t_1\in (0,1]$, the allocation rule $(q^1_1(t_1,0), q^2_1(t_1,0))$ is constant $(1,0)$, but $pay^1_1(t_1,0)$ is not.
In other words, in this interval, the agent wants to report a type between $(1/4,3/4)$, where the payment is zero.
\begin{theorem}
No DIC and {\em ex post} IR auction achieves the revenue bound $\frac{29}{24}$.
\label{noticir}
\end{theorem}
\begin{proof}
Suppose by contradiction that there exists a such mechanism.
Since $u(\frac{1}{2})=0$, {\em ex post} IR implies that
$u_1(\frac{1}{2},t_2)=0$ for type $t_2\in[0,1]$.
In particular, $u_1(\frac{1}{2},0)=0$.
When $t_2=0$, DIC constraint implies that
$\partial u(t_1,0)/\partial t_1 =q^1(t_1,0)-q^2(t_1,0)=1,~t_1\in(0,1]$.
Thus $u(t_1,0)=t_1-\frac{1}{2},~t_1\in(0,1]$.
However, {\em ex post} IR also requires $u(t_1,0)\geq 0$, contradicting the DIC constraint.
\end{proof}

\section{Acknowledgement}
The authors would like to thank Weiran Shen, Yulong Zeng, and Song Zuo for helpful discussion. Part of this work has been done while Pingzhong Tang was a participant of the Economics and Computation program, Simon Institute, UC Berkeley. He benefits from discussions with Nima Haghpanah and Christos Tzamos on the single-bidder case.

\bibliographystyle{ACM-Reference-Format-Journals}
\bibliography{acmsmall-sample-bibfile}
\newpage
\appendix
\section{Appendix}

\subsection{Proof of Lemma~\ref{monoborder}}
\begin{proof}
For the first part, the proof is exact same  to Lemma~\ref{generalborder} except using Border's theorem in symmetric version instead.
For the second part, when the bound is achieved, the inequality (\ref{eq0}) becomes equality. In symmetric version, it is
$$\int_{x(t)\geq v}q(t)f(t)dt=\frac{1-(1-\int_{x(t)\geq v}f(t)dt)^n}{n}, \forall v$$
Let $\mathcal{W}=\{s: x(s)\geq \min_{t\in \mathcal{C}}\{x(t)\}\}$, $w=\int_{t\in \mathcal{W}}f(t)dt$, $c=\int_{t\in \mathcal{C}}f(t)dt$, and $d=\int_{t\in \mathcal{D}}f(t)dt$
\begin{eqnarray}
\textrm{Then}\quad&&\int_{\mathcal{W}}q(t)f(t)dt=\frac{1-(1-\int_{\mathcal{W}}f(t)dt)^n}{n}=\frac{1-(1-w)^n}{n}\label{eqa1}\\
&&\int_{\mathcal{W}-\mathcal{C}}q(t)f(t)dt\leq\frac{1-(1-\int_{\mathcal{W}-\mathcal{C}}f(t)dt)^n}{n}=\frac{1-(1-w+c)^n}{n}\label{eqa2}\\
&&\int_{\mathcal{W}+\mathcal{D}}q(t)f(t)dt\leq\frac{1-(1-\int_{\mathcal{W}+\mathcal{D}}f(t)dt)^n}{n}=\frac{1-(1-w-d)^n}{n}\label{eqa3}
\end{eqnarray}
By (\ref{eqa1})-(\ref{eqa2}) and (\ref{eqa3})-(\ref{eqa1}), we have
\begin{eqnarray}
(\ref{eqa1})-(\ref{eqa2})\quad \int_{\mathcal{C}}q(t)f(t)dt&\geq&\frac{(1-w+c)^n-(1-w)^n}{n}\nonumber\\
\frac{\int_{\mathcal{C}}q(t)f(t)dt}{c}&\geq&\frac{1}{n}\sum_{k=0}^{n-1}(1-w+c)^(1-w)^{n-1-k}\nonumber\\
(\ref{eqa3})-(\ref{eqa1})\quad \int_{\mathcal{D}}q(t)f(t)dt&\leq&\frac{(1-w)^n-(1-w-d)^n}{n}\nonumber\\
\frac{\int_{\mathcal{D}}q(t)f(t)dt}{d}&\leq&\frac{1}{n}\sum_{k=0}^{n-1}(1-w-d)^k(1-w)^{n-1-k}\nonumber
\end{eqnarray}
Hence
$\frac{\int_{\mathcal{C}}q(t)f(t)dt}{c}\geq
\frac{\int_{\mathcal{D}}q(t)f(t)dt}{d}$.
\end{proof}

\subsection{Proof of Lemma~\ref{derivative}}
\begin{proof}
Since $H^{ir}(z)$ is convex, so $\frac{\partial H^{ir}}{\partial z}(z)$ is weakly increasing.
We have
$$\frac{\partial H^{ir}}{\partial z}(z)=(\int^{F^{-1}(z)}_{0}h^{ir}(s)ds)'=h^{ir}(F^{-1}(z))(F^{-1}(z))'=\frac{h^{ir}(F^{-1}(z))}{f(F^{-1}(z))}$$
so $\frac{h^{ir}(t)}{f(t)}$ is weakly increasing as well.
\end{proof}

\subsection{Proof of Lemma~\ref{abovelemma}}
\begin{proof}
By definition,
\begin{eqnarray}
H(z,t^*)&=&\int^{F^{-1}(z)}_{0} h(s,t^*) ds\nonumber\\
\frac{\partial H}{\partial z}(z,t^*)&=&\frac{h(F^{-1}(z))}{f(F^{-1}(z))}\nonumber\\
\frac{\partial H}{\partial z}(F(t),t^*)&=&\frac{h(t)}{f(t)}\nonumber
\end{eqnarray}
Since
\begin{displaymath}
h(t,t^*) = \left\{ \begin{array}{ll}
tf(t)+F(t) & t\leq t^*\\
tf(t)+F(t)-1 & t> t^*
\end{array} \right.
\end{displaymath}
we have
\begin{displaymath}
\frac{\partial_- H}{\partial z}(F(t^*),t^*) = t^*+\frac{F(t^*)}{f(t^*)}>t^*+\frac{F(t^*)-1}{f(t^*)}=\frac{\partial_+ H}{\partial z}(F(t^*),t^*)
\end{displaymath}
So point $(F(t^*),H(F(t^*),t^*))$  is above the convex hull.
\end{proof}

\subsection{Proof of Lemma~\ref{biglemma}}
\begin{proof}
We prove the right direction for $q^1$, we consider the case $t^*\in(0,b)$  first.
A small lemma is required,
\begin{lemma}
$\exists \epsilon>0,~\forall t_2\in(t^*,t^*+\epsilon)$, we have $l_1^{max}(t_2,t_2)<F(t^*)$.
\label{lemma1}
\end{lemma}

Second we prove when $t_2$ is close to $t^*$, $l_1^{min}(t_2,t_2)$ is close to $l_1^{max}(t^*,t^*)$.
\begin{lemma}
$\exists \epsilon>0,~\forall t_2\in(t^*,t^*+\epsilon)$, we have $l_1^{min}(t_2,t_2)\geq l_1^{max}(t^*,t^*)$.
\label{lemma2}
\end{lemma}

\begin{lemma}
$\exists \epsilon>0,~\forall t_2\in(t^*,t^*+\epsilon)$, we have $l_1^{min}(t_2,t_2)\leq l^{max}_1(t^*,t^*)+\epsilon$.
\label{lemma3}
\end{lemma}

Third we prove when $t_2$ is close to $t^*$, $l_2^{min}(t_2,t_2)$ is close to $l_2^{max}(t^*,t^*)$.
\begin{lemma}
$\exists \epsilon>0,~\forall t_2\in(t^*,t^*+\epsilon)$, we have $l_2^{min}(t_2,t_2)\geq l^{max}_2(t^*,t^*)$
\label{lemma4}
\end{lemma}

\begin{lemma}
$\exists \epsilon>0,~\forall t_2\in(t^*,t^*+\epsilon)$, we have $l_2^{min}(t_2,t_2)\leq l_2^{max}(t^*,t^*)+\epsilon$.
\label{lemma5}
\end{lemma}

Combined above lemmas,
$\exists \epsilon>0,~\forall t_2\in(t^*,t^*+\epsilon)$
we have $l^{min}_1(t_2,t_2)\in[l^{max}_1(t^*,t^*),l^{max}_1(t^*,t^*)+\epsilon]$ and
$l^{min}_2(t_2,t_2)\in[l^{max}_2(t^*,t^*),l^{max}_2(t^*,t^*)+\epsilon]$.
By the definition of $q^1$, $\hat{q}^1(t_2,t_2)$ can be arbitrarily close to $\check{q}^1(t^*,t^*)$.
Also, $\hat{q}^1(t,t)$ is increasing in $l_1$ and $l_2$, so $\hat{q}^1(t_2,t_2)\geq \check{q}^1(t^*,t^*)$.

When $t^*=0$, the formula of the allocation rule might be different.
But using the exact same method, we can prove $q^1(0,0)=\lim_{t^*\rightarrow 0}q^1(t^*,t^*)$.

The left direction is similar. It is also similar for $q^2$.
To sum up the value intervals seamlessly connect as $t^*$ increases.
\end{proof}

\subsection{Proof of Lemma~\ref{lemma1}}
\begin{proof}
\begin{figure}[h]
  \centering
  \includegraphics[width=6cm]{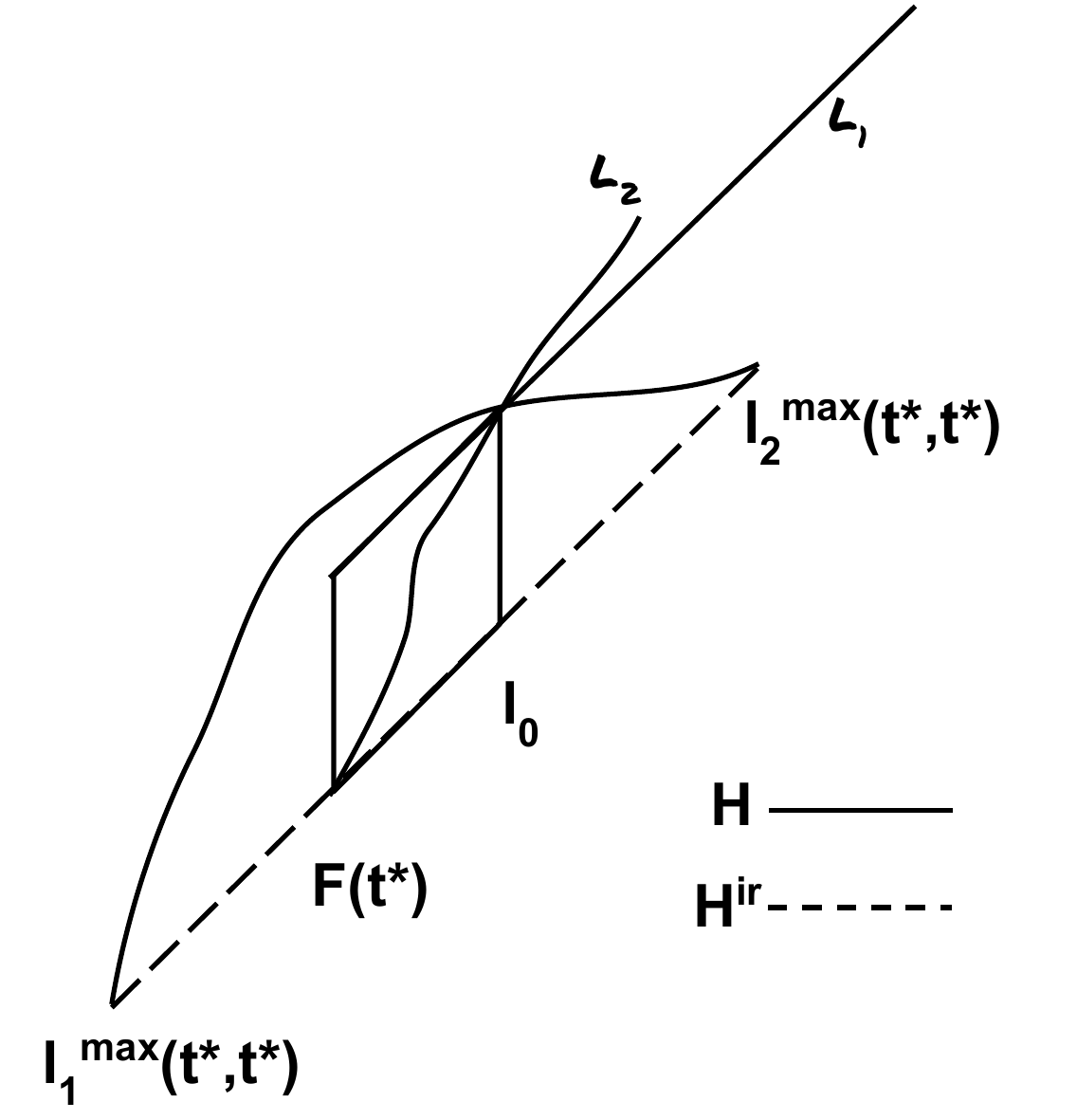}\\
  \caption{Path $\mathcal{L}_2$ and the parallelogram}\label{small}
\end{figure}

Look at the Fig \ref{small}, we put a parallelogram on point $(F(t^*), H^{ir}(F(t^*),t^*))$. One edge is vertical and one edge is on the boundary of
$H^{ir}(z,t^*)$. The opposite point moves along the path
$\mathcal{L}_2(z,F^{-1}(z)-t^*+H^{ir}(z,t^*)),~z>F(t)^*$. We consider the first parallelogram
touches the curve $H(z,t^*)$.
Then the axis of the opposite point is
$$l_0=\max_x\{  \forall t\in[t^*,F^{-1}(x)]~ H(F(t),t^*)\geq H^{ir}(F(t),t^*)+t-t^*\}$$
We claim when $t_2\in(t^*, F^{-1}(l_0)]$, $l_{1}^{max}(t_2,t_2)<F(t^*)$.

Otherwise, suppose $l_1^{max}(t_2,t_2)\geq F(t^*)$. Since $t_2>t^*$, we have
$$\frac{\partial H^{ir}}{\partial z}(F(t_2),t_2)\geq \frac{\partial H^{ir}}{\partial z}(F(t^*),t_2)>\frac{\partial H^{ir}}{\partial z}(F(t^*),t^*)$$

Point $A(l^{max}_1(t_2,t_2),H(l_1^{max}(t_2,t_2),t_2))$ and
$B(l_2^{max}(t^*,t^*),H(l_2^{max}(t^*,t^*),t_2))$.
Let $\mathcal{L}_1$ denote the straight line which contains one of the edge of the parallelogram.
Then point $A$ is above the line $\mathcal{L}_1$
\begin{eqnarray}
&&H(l_2^{max}(t^*,t^*),t_2))\nonumber\\
&=&H(l_2^{max}(t^*,t^*),t^*)+t_2-t^*\nonumber\\
&\leq&H(l_2^{max}(t^*,t^*),t^*)+F^{-1}(l_0)-t^*\nonumber
\end{eqnarray}
By definition of $l_0$, we have
$H(l_0,t^*)=H^{ir}(l_0,t^*)+F^{-1}(l_0)-t^*$. So
\begin{eqnarray}
&&H(l_2^{max}(t^*,t^*),t_2))\nonumber\\
&\leq&H(l_2^{max}(t^*,t^*),t^*)+H(l_0,t^*)-H^{ir}(l_0,t^*)\nonumber
\end{eqnarray}
This means point $B$ is on or below line $\mathcal{L}_1$.
For simplicity, we use $(AB)'$ and $(\mathcal{L}_1)'$ to denote the gradient of line $AB$ and line $\mathcal{L}_1$ respectively.
Then
\begin{eqnarray}
&&(AB)'\leq (\mathcal{L}_1)'\nonumber\\
&\Rightarrow& \frac{\partial H^{ir}}{\partial z}(F(t_2),t_2)\leq \frac{\partial H^{ir}}{\partial z}(F(t^*),t^*)\nonumber\\
&\Rightarrow& \frac{\partial H^{ir}}{\partial z}(F(t_2),t^*)\leq \frac{\partial H^{ir}}{\partial z}(F(t^*),t^*)\nonumber
\end{eqnarray}
Contradiction.
\end{proof}

\subsection{Proof of Lemma~\ref{lemma2}}
\begin{proof}
First prove a gadget,
\begin{eqnarray}
&&H(z,t_2)=H(z,t^*) \qquad z\in[0,F(t^*)]\nonumber\\
&&H(z,t_2)>H(z,t^*) \qquad z\in(F(t^*),1]\nonumber\\
&\Rightarrow&  \frac{\partial H^{ir}}{\partial z}(F(t^*),t_2)> \frac{\partial H^{ir}}{\partial z}(F(t^*),t^*)\label{eqz01}
\end{eqnarray}
\begin{eqnarray}
&& \frac{\partial_{+} H^{ir}}{\partial z}(l_1^{min}(t_2,t_2),t_2)\nonumber\\
&=&\frac{\partial H^{ir}}{\partial z}(F(t_2),t_2)\nonumber\\
&\geq&\frac{\partial H^{ir}}{\partial z}(F(t^*),t_2)\nonumber\\
&>&\frac{\partial H^{ir}}{\partial z}(F(t^*),t^*)\qquad{\textrm {by} (\ref{eqz01}})\nonumber\\
&=&\frac{\partial_{+} H^{ir}}{\partial z}(l^{max}_1(t^*,t^*),t^*)\nonumber\\
&\geq&\frac{\partial_{-} H^{ir}}{\partial z}(l^{max}_1(t^*,t^*),t^*)\nonumber\\
&=&\frac{\partial_{-} H^{ir}}{\partial z}(l^{max}_1(t^*,t^*),t_2)\quad (H^{ir}\textrm{~is the covex hull.})\nonumber
\end{eqnarray}
Thus $l^{min}_1(t_2,t_2)\geq l^{max}_1(t^*,t^*)$.
\end{proof}

\subsection{Proof of Lemma~\ref{lemma3}}
\begin{proof}
Consider the following three points:
\begin{eqnarray}
&&A(l^{max}_1(t^*,t^*), H(l^{max}_1(t^*,t^*),t^*)))\nonumber\\
&&B(l^{min}_1(t_2,t_2), H(l^{min}_1(t_2,t_2),t_2)))\nonumber\\
&&C(l^{max}_2(t^*,t^*), H(l^{max}_2(t^*,t^*),t_2)))\nonumber
\end{eqnarray}
Points $A$,$B$ and $C$ all lie on curve $(z,H^{ir}(z,t_2))$. By convexity, we know
$$(AB)'\leq (AC)'$$

Let $M(z,t^*),~z\in[0,F(t^*)]$ be the largest convex function under $H(z,t^*),~z\in[0,F(t^*)]$.
Note that any $t_2\in(t^*,F^{-1}(l^{min}_2(t^*,t^*)))$, $H(z,t_2)=H(z,t^*),~z\in[0,F(t^*)]$.
We must have $M(z,t^*)\geq H^{ir}(z,t_2),~z\in[0,F(t^*)]$.

By Lemma~\ref{lemma1}, we can pick $t_2$ which is small enough so that $l^{max}_1(t_2,t_2)<F(t^*)$. Then
$M(l^{min}_1(t_2,t_2),t^*)\geq H(l^{min}_1(t_2,t_2),t_2)$.
By definition of $M$,
\begin{eqnarray}
&&H(l^{min}_1(t_2,t_2),t^*)\nonumber\\
&\geq&M(l^{min}_1(t_2,t_2),t^*)\nonumber\\
&\geq&H(l^{min}_1(t_2,t_2),t_2)\nonumber\\
&=&H(l^{min}_1(t_2,t_2),t^*)\nonumber
\end{eqnarray}
Thus all the inequalities become equalities. Point $B$ must lie on curve $M(z,t^*)$.

Consider point $D(l^{max}_1(t^*,t^*)+\epsilon, M(l^{max}_1(t^*,t^*)+\epsilon,t^*))$, we must have
$$(AD)'\leq (AB)'$$, furthermore $(AD)'\leq (AC)'$.

$$\frac{H^{ir}(l^{max}_1(t^*,t^*)+\epsilon, t^*) - H(l^{max}_1(t^*,t^*), t^*)}{\epsilon}\leq
\frac{H(l^{max}_2(t^*,t^*), t_2) - H(l^{max}_1(t^*,t^*), t^*)}{l_2^{max}(t^*,t^*)-l^{max}_1(t^*,t^*)}$$
We transform the inequation and leave only $H(l^{max}_2(t^*,t^*), t_2)$ on the right hand side.
\begin{eqnarray}
W_1(\epsilon,t^*)&\leq& H(l^{max}_2(t^*,t^*), t_2)\nonumber\\
W_1(\epsilon,t^*)&\leq& H(l^{max}_2(t^*,t^*), t^*)+t_2-t^*\nonumber
\end{eqnarray}
We transform the inequation and leave only $t_2$ on the right hand side.
$$W_2(\epsilon,t^*)\leq t_2$$

Where $W_1$ and $W_2$ are functions of $\epsilon$ and $t^*$.
So as long as $t_2<W_2(\epsilon, t^*)$, we have
$l^{min}_1(t_2,t_2)\leq l^{max}_1(t^*,t^*)+\epsilon$

Consider point $E(l^{max}_2(t^*,t^*),H(l^{max}_2(t^*,t^*),t^*))$, we have $(AD)'>(AE)'$, i.e.
$$\frac{H^{ir}(l^{max}_1(t^*,t^*)+\epsilon, t^*) - H(l^{max}_1(t^*,t^*), t^*)}{\epsilon}>
\frac{H(l^{max}_2(t^*,t^*), t^*) - H(l^{max}_1(t^*,t^*), t^*)}{l_2^{max}(t^*,t^*)-l^{max}_1(t^*,t^*)}$$
Then we get $W_2(\epsilon, t^*)>t^*$. So as long as
$t_2\in (t^*, W_2(\epsilon, t^*))$, we have $l^{min}_1(t_2,t_2)\leq l^{max}_1(t^*,t^*)+\epsilon$.
\end{proof}

\subsection{Proof of Lemma~\ref{lemma4}}
\begin{proof}
By Lemma~\ref{lemma1}, we can pick $t_2>t^*$ which is small enough so that
\begin{eqnarray}
F(t_2)<l^{min}_2(t^*,t^*)\nonumber\\
l^{min}_1(t_2,t_2)<F(t^*)\nonumber
\end{eqnarray}
We claim that $l^{min}_2(t_2,t_2)\geq l^{max}_2(t^*,t^*)$.
Suppose not, consider the following six points in Fig~\ref{sixpoints}.
\begin{eqnarray}
A(l^{max}_1(t^*,t^*), H(l^{max}_1(t^*,t^*),t^*))\nonumber\\
B(l^{min}_1(t_2,t_2), H(l^{min}_1(t_2,t_2),t^*))\nonumber\\
C(l^{min}_2(t_2,t_2), H(l^{min}_2(t_2,t_2),t_2))\nonumber\\
D(l^{max}_2(t^*,t^*), H(l^{max}_2(t^*,t^*),t_2))\nonumber\\
E(l^{min}_2(t_2,t_2), H(l^{min}_2(t_2,t_2),t^*))\nonumber\\
F(l^{max}_2(t^*,t^*), H(l^{max}_2(t^*,t^*),t^*))\nonumber
\end{eqnarray}
\begin{figure}
  \centering
  \includegraphics[width=9cm]{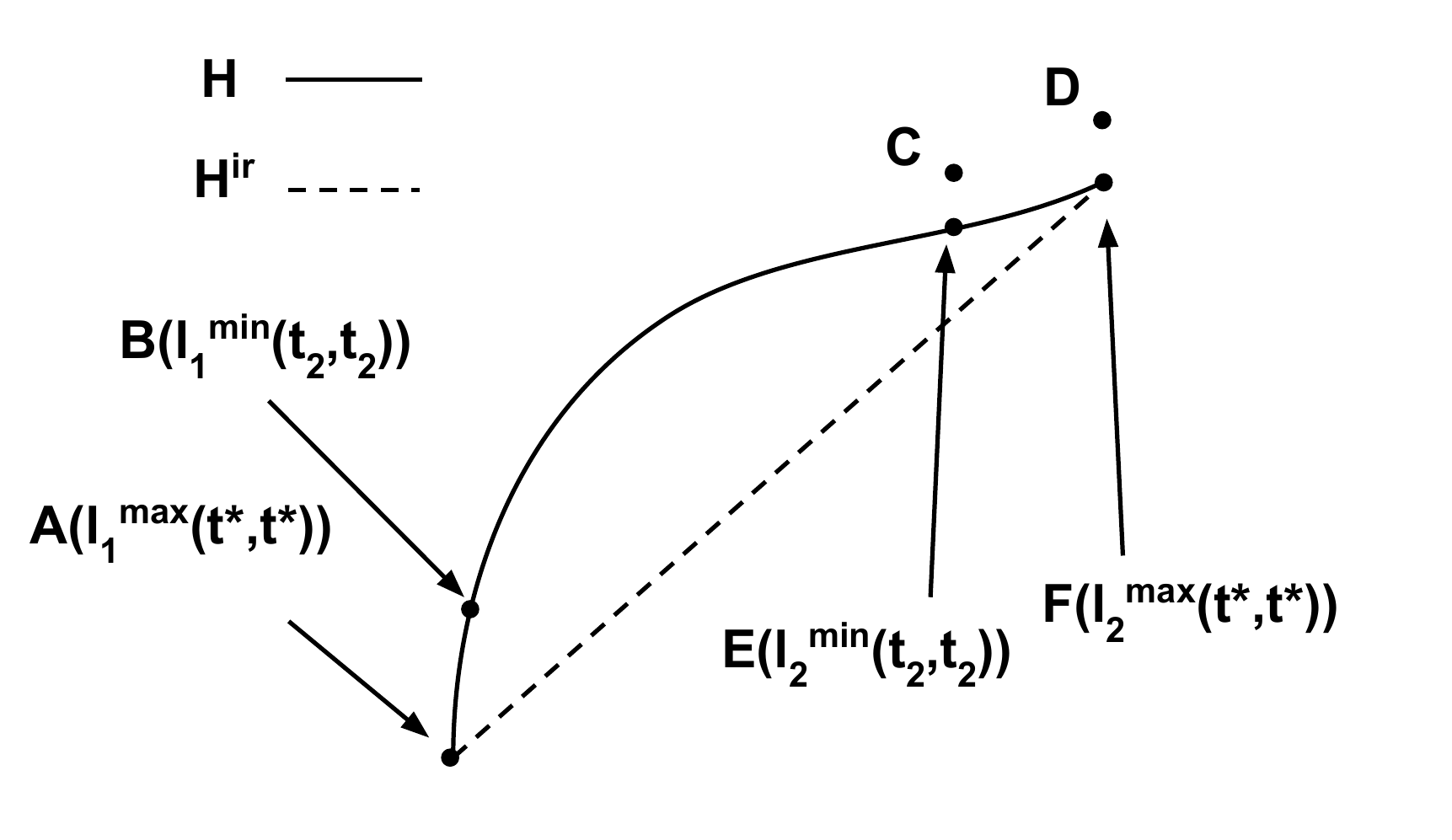}\\
  \caption{Six points}\label{sixpoints}
\end{figure}

Because $l^{min}_2(t_2,t_2)<l^{max}_2(t^*,t^*)$, point $C$ is on the left of point $D$.
Since $l^{min}_2(t_2,t_2)>F(t_2)$ and $l^{max}_2(t^*,t^*)>F(t_2)$, we have
$$H(z,t_2)-H(z,t^*)=F(t_2)-F(t^*),~z>F(t_2)$$
Thus $CDFE$ is parallelogram.
\begin{eqnarray}
&&\frac{\partial H^{ir}(F(t^*),t^*)}{\partial z}=(AF)'\nonumber\\
&\geq&(EF)'=(CD)'\nonumber\\
&\geq&(BC)'=\frac{\partial H^{ir}(F(t^*),t_2)}{\partial z}\nonumber\\
&>&\frac{\partial H^{ir}(F(t^*),t^*)}{\partial z}\nonumber
\end{eqnarray}
Contradiction.
So as long as $t_2$ is small enough such that
$F(t_2)<l^{min}_2(t^*,t^*)$ and $l^{min}_1(t_2,t_2)<F(t^*)$,
we have $l^{min}_2(t_2,t_2)\geq l^{max}_2(t^*,t^*)$.
\end{proof}
\subsection{Proof of Lemma~\ref{lemma5}}
\begin{proof}
For any $\epsilon>0$,
consider the boundary straight line on which point $(l_2^{max}(t^*,t^*)+\epsilon, H^{ir}(l^{max}_2(t^*,t^*)+\epsilon, t^*))$ lies.
This line might be a point or a line segment.
Let point $(l_3, H^{ir}(l_3,t^*))$ denote the left endpoint of this line. If this line is a point, then $l_3$ equals $l_2^{max}(t^*,t^*)+\epsilon$.
Formally,
$$l_3=\min_z\{\frac{\partial_+ H^{ir}}{\partial z}(z,t^*) =\frac{\partial_+ H^{ir}}{\partial z}(l^{max}_2(t^*,t^*)+\epsilon, t^*)\}$$
Then $l_3\geq l^{max}_2(t^*,t^*)$.
Choose $t_2\in (t^*, F^{-1}(l^{min}_2(t^*,t^*)))$, if $l^{min}_2(t_2,t_2)>l_3$, consider the following five points:
\begin{eqnarray}
&&A(l^{min}_1(t_2,t_2), H(l^{min}_1(t_2,t_2),t_2))\nonumber\\
&&B(l_3,H(l_3,t_2))\nonumber\\
&&C(l^{min}_2(t_2,t_2),H(l^{min}_2(t_2,t_2),t_2))\nonumber\\
&&D(l_3,H(l_3,t^*))\nonumber\\
&&E(l^{min}_2(t_2,t_2), H(l^{min}_2(t_2,t_2),t^*))\nonumber
\end{eqnarray}

Since $l^{min}_2(t_2,t_2)>F(t_2)$, and $l_3>F(t_2)$, we have
$$H(z,t_2)-H(z,t^*)=F(t_2)-F(t^*)$$
Then $BCED$ is a parallelogram.

By the definition of points $A$ and $C$, we know point $B$ is above the line $AC$. So
$(BA)'>(AC)'>(BC)'$.
$$(BC)'=(DE)'\geq \frac{\partial_+ H^{ir}}{\partial z}(l_3, t^*)=\frac{\partial_+ H^{ir}}{\partial z}(l_2^{max}(t^*,t^*)+\epsilon, t^*)$$

By Lemma~\ref{lemma1}, we can choose $t_2$ small enough such that
$l^{min}_1(t_2,t_2)<F(t^*)$, then
$$H(l^{min}_1(t_2,t_2),t_2)=H(l^{min}_1(t_2,t_2),t^*)$$

Thus, we have
\begin{eqnarray}
(BA)'&=&\frac{H(l_3,t_2)-H(l^{min}_1(t_2,t_2),t_2)}{l_3-l^{min}_1(t_2,t_2)}\nonumber\\
&=&\frac{F(t_2)-F(t^*)+H(l_3,t^*)-H(l^{min}_1(t_2,t_2),t^*)}{l_3-l^{min}_1(t_2,t_2)}\quad\textrm{(BCDE is parallelogram)}\nonumber
\end{eqnarray}

By $(BA)'>(BC)'$, we have
$$\frac{F(t_2)-F(t^*)+H(l_3,t^*)-H(l^{min}_1(t_2,t_2),t^*)}{l_3-l^{min}_1(t_2,t_2)}\geq \frac{\partial_+ H^{ir}}{\partial z}(l^{max}_2(t^*,t^*)+\epsilon, t^*)$$
\begin{eqnarray}
F(t_2)-F(t^*)&\geq& [l_3-l^{min}_1(t_2,t_2)]\cdot \frac{\partial_+ H^{ir}}{\partial z}(l^{max}_2(t^*,t^*)+\epsilon, t^*) \nonumber\\
&&+H(l^{min}_1(t_2,t_2),t^*)-H(l_3,t^*)\nonumber\\
&\geq&\min_{z\leq F(t^*)}\{(l_3-z)\cdot \frac{\partial_+ H^{ir}}{\partial z}(l^{max}_2(t^*,t^*)+\epsilon, t^*) \nonumber\\
&&+H(z,t^*)-H(l_3,t^*)\}\nonumber\\
&=&W(\epsilon,t^*)\nonumber
\end{eqnarray}
We let $W(\epsilon,t^*)$ denote the minimum part. Hence if $F(t_2)-F(t^*)<W(\epsilon,t^*)$, we have
$$l^{min}_2(t_2,t_2)\leq l_3\leq l^{max}_2(t_2,t_2)+\epsilon$$

Next we prove $W(\epsilon,t^*)>0$.
By definition of $l_3$, for any $z\in[0,l_3)$, point $(z,H(z,t^*))$ is above the straight line that goes through point $(l_3,H^{ir}(l_3,t^*))$ with gradient
$g=\frac{\partial_+ H^{ir}}{\partial z}(l^{max}_2(t^*,t^*)+\epsilon, t^*)$. So
$$H(z,t^*)>H(l_3,t^*)+(z-l_3)g$$
So $W(\epsilon,t^*)\geq 0$. Suppose $W(\epsilon,t^*)=0$. Since
$(l_3-z)g+H(z,t^*)-H(l_3,t^*)$ is a continuous function of $z$, then
there exists $z_0$ such that $(l_3-z_0)g+H(z_0,t^*)-H(l_3,t^*)=0$, which means $(z_0,H(z_0,t^*))$ lies on the straight line
that goes through point $(l_3,H^{ir}(l_3,t^*))$ with gradient $g$. Contradiction.
Hence $W(\epsilon,t^*)>0$. To sum up, as long as $t_2<F(t^*)+W(\epsilon,t^*)$, we have
$l_1^{min}(t_2,t_2)\leq l_2^{max}(t^*,t^*)+\epsilon$.
\end{proof}

\end{document}